\documentclass[pra,floatfix,amsmath,twocolumn]{revtex4}
\usepackage{amssymb}
\usepackage{graphicx}
\usepackage{graphics}
\usepackage{amsmath}
\usepackage{amsthm}

\def\x{X}
\def\y{Y}
\def\z{Z}

\def\i{{\bf i}}

\def\j{{\bf j}}

\def\k{{\bf k}}
\def\l{{\bf l}}

\newcommand{\bea}{\begin{eqnarray}}
\newcommand{\eea}{\end{eqnarray}}
\def\bi{\begin{itemize}}
\def\ei{\end{itemize}}
\def\bc{\begin{center}}
\def\ec{\end{center}}

\def\C{\hbox{$\mit I$\kern-.7em$\mit C$}}
\def\R{\hbox{$\mit I$\kern-.6em$\mit R$}}

\def\ket#1{|#1\rangle}
\newcommand{\one}{\mbox{$1 \hspace{-1.0mm}  {\bf l}$}}
\def\tr{\mathrm{tr}}
\def\ket#1{\left| #1\right>}
\def\bra#1{\left< #1\right|}

\newcommand{\proj}[1]{\ket{#1}\bra{#1}}

\newtheorem{theorem}{Theorem}
\newtheorem{corollary}[theorem]{Corollary}
\newtheorem{lemma}[theorem]{Lemma}
\begin{document}

\title{Local unitary equivalence and entanglement of multipartite pure states}

\author{B. Kraus}

\affiliation{Institute for Theoretical Physics, University of
Innsbruck, Austria}

\begin{abstract}

The necessary and sufficient conditions for the equivalence of
arbitrary $n$--qubit pure quantum states under Local Unitary (LU)
operations derived in [B. Kraus Phys. Rev. Lett. {\bf 104}, 020504
(2010)] are used to determine the different LU--equivalence classes
of up to five--qubit states. Due to this classification new
parameters characterizing multipartite entanglement are found and their physical interpretation
is given. Moreover, the method is used to derive examples of two
$n$--qubit states (with $n>2$ arbitrary) which have the properties
that all the entropies of any subsystem coincide, however, the
states are neither LU--equivalent nor can be mapped into each other
by general local operations and classical communication.
\end{abstract}
\maketitle

\section{Introduction}

The subtle properties of multipartite entangled states
allow for many fascinating applications of quantum
information, like one--way quantum computing, quantum error
correction, and quantum secret sharing \cite{Gothesis97,RaBr01}.
The theory of many--body states plays also an
important role in other fields of physics which deal with many-body systems \cite{AmFa08}. Thus, the investigation of the non--local properties of quantum states is at the heart of quantum information theory. Compared to the bipartite case, which is well understood, the multipartite case is much more complex due to the exponential grows in the dimension of the Hilbert space. Despite its relevance and the enormous effort of theorists, many problems regarding multipartite entanglement are still unsolved \cite{HoHo07}. Several entanglement measures for
multipartite states, like the tangle \cite{CoKu00}, the Schmidt measure \cite{Sm}, the localizable entanglement \cite{VePoCi04}, or
geometric measure of entanglement \cite{GM} have been introduced. Moreover, different
classes of entangled states have been identified \cite{DuViCi00},
and a normal form of multipartite states has been presented
\cite{Ves}.
However, even for the simplest case of three qubits the
entanglement properties are still not completely understood. One of the main reasons for the lack of knowledge is, arguably, that for many--body entangled states we do have only few applications \cite{HoHo07}. This
results into the existence of few known operational entanglement
measures.

One approach to gain insight into the complicated structure of multipartite states is to consider a restricted class of states, like for instance stabilizer states \cite{Gothesis97}, matrix--product states \cite{Vidal}, projected entangled pair states \cite{PVC06}, Locally Maximally Entangleable States (LMESs) \cite{KrKr08}, or Gaussian state \cite{BrLo05}. Considering a restricted set of states enabled researchers to gain a
lot of intuition about the usefulness and manipulation of
them. This knowledge in turn led to many of the
fascinating applications of multipartite states.

Another way to gain insight into the entanglement properties of quantum
states is to consider their interconvertability. That is, given two
states $\ket{\Psi}$, $\ket{\Phi}$ the question is whether or not
$\ket{\Psi}$ can be transformed into $\ket{\Phi}$ by local
operations \cite{HoHo07}. One particularly interesting case is the LU-equivalence of
multipartite states. We say that a $n$--partite state, $\ket{\Psi}$
is LU--equivalent to $\ket{\Phi}$ ($\ket{\Psi}\simeq_{LU}
\ket{\Phi}$) if there exist local unitary operators, $U_1,\ldots,
U_n$, such that $\ket{\Psi}=U_1\otimes \cdots \otimes U_n
\ket{\Phi}$. Note that two states which are LU--equivalent are
equally useful for any kind of application and they posses
precisely the same amount of entanglement. Another insight is gained by considering more general operations, like (deterministic) local operations and classical communication (LOCC). Since the implementation of such operations does not consume entanglement, a state, $\ket{\Psi}$ which can be mapped into the state $\ket{\Phi}$ by LOCC is necessarily at least as entangled as $\ket{\Phi}$. Thus, all these investigations of convertibility lead to a new insight into general problem of
characterizing the different types of entangled quantum states.

In this article we will mainly focus on the LU--equivalence of multipartite states. Local polynomial invariants have been introduced to distinguish the different LU--equivalence classes \cite{GrRo98}. However, even though it is known that it is
sufficient to consider only a finite set of them, this complete
finite set is known only for very few simple cases. In \cite{Kr09} a method to solve the LU--equivalence problem for
arbitrary $n$--qubit states has been presented. There, an algorithm which determines the local unitaries, which map the states into each other (if they exist) has been derived. Within this algorithm, different classes of states, which are easily characterized, are distinguished. It has been shown that two states which are within two different classes cannot have the same entanglement.

Here, we will use the algorithm presented in \cite{Kr09} in order to investigate the non--local properties of multipartite states. We will present the LU--equivalence classes of few--body systems and obtain a new insight into multipartite entanglement. The main results derived here will be summarized in Sec \ref{Secresults}.

The sequel of the paper is organized as follows. After presenting  the main results of this article (Sec \ref{Secresults}), we review the necessary and
sufficient conditions for LU--equivalence derived in \cite{Kr09}.
In Section \ref{SecAdd} we will derive some additional methods to determine the local unitaries (if they exist) which interconvert the two states. In Section \ref{secExamples} we will characterize the LU-equivalence classes of up to 4
qubits. For five--qubit states we consider the most challenging  class (for using the algorithm) and show how the local unitaries can be determined then. For an arbitrary $n$ (with $n>2$) the existence of $n$--qubit states, $\ket{\Psi}$ which are not LU--equivalent to their complex conjugate will be shown by presenting examples in Section \ref{SecLOCC}. In Section \ref{SecMixed} it will be shown how the algorithm can be employed to solve the LU--equivalence problem for certain mixed states and also states which describe $d$--level systems. The new insight gained into multipartite entanglement will be discussed in Section \ref{SecEntanglement}.

Throughout this paper the following notation is used. The Pauli operators will mainly be denoted by $X,Y,Z$. Whenever we need the whole set of Pauli operators we will use the notation $\Sigma_1=X,\Sigma_2=Y,$ and $,\Sigma_3=Z$ and $H$ denotes the Hadamard
transformation. Otherwise, the subscript of an operator will always denote the
system it is acting on, or the system it is describing. The reduced
states of system $i_1,\ldots i_k$ of $\ket{\Psi}$ ($\ket{\Phi}$)
will always be denoted by $\rho_{i_1\ldots i_k}$
($\sigma_{i_1\ldots i_k}$) resp., i.e. $\rho_{i_1\ldots
i_k}=\tr_{\neg i_1\ldots \neg i_k}(\proj{\Psi})$. We denote by
${\bf i}$ the classical bit--string $(i_1,\ldots, i_n)$ with
$i_k\in\{0,1\}$ $\forall k\in \{1,\ldots, n\}$ and $\ket{{\bf
i}}\equiv\ket{i_1,\ldots, i_n}$ denotes the computational basis.
Normalization factors as well as the tensor product symbol will be
omitted whenever it does not cause any confusion and $\one$ will
denote the normalized identity operator. The eigenvalues of some
matrix $M$ will be denoted by $\mbox{eig}(M)$. For a subsystem $A$ we will denote by $E_A(\ket{\Psi})$ the
bipartite entanglement between $A$ and the remaining
systems measured with the Von Neumann entropy of the reduced state,
$\rho_A$. For instance, $E_i(\ket{\Psi})=S(\rho_i)$ will denote the
entanglement between qubit $i$ and the remaining $n-1$ qubits.
As commonly used, the states $\ket{\Phi^{\pm}}=\ket{00}\pm \ket{11}$, $\ket{\Psi^{\pm}}=\ket{01}\pm \ket{10}$ denote the Bell basis. The state $\ket{\Psi^\ast}$ will denote the complex conjugate of the state $\ket{\Psi}$ in the computational basis and $M^T$ will denote the transpose of an operator $M$ in the computational basis.


\section{Main results}

\label{Secresults}

On the one hand, the criterion for LU--equivalence \cite{Kr09} will be used to characterize the LU--equivalence classes for few--body states. On the other hand, it will be employed to shed new light on multipartite entanglement. Furthermore, the criterion will be generalized to certain mixed states and also states which describe $d$--level systems. The main results derived here can be summarized as follows.

i) Characterization of LU--equivalence classes: The LU--equivalence classes of quantum states describing up to five qubits will be characterized. For two, three and four--qubit states all the classes will be explicitly derived. For five--qubit states, whose classification would work analogously, only the most challenging subset of states will be considered. It will be explicitly shown how the algorithm can be used to determine the local unitaries (if they exist) which transform one state into the other.

ii) New insight into multipartite entanglement:

a) The algorithm presented in \cite{Kr09} distinguishes between different classes of states, like the class of states with $\rho_{12}=\rho_1\otimes \one_2$ and the one with $\rho_{12}\neq \rho_1\otimes \one_2$. It will be shown how this classification enables us to gain a new insight into multipartite entanglement. For instance, one class would be the one where one of the two--qubit reduced states is completely mixed. For four--qubit states it will be shown that this class is completely characterized by only three non--local parameters. Moreover, it will be proven that two states within this class are LU--equivalent iff the corresponding sets of three parameters coincide. Naturally, all the entanglement contained in a state within this class is
also determined by those three parameters to which the following operational meaning can be given. Recall that the completely positive map (CPM), ${\cal E}_\Psi$, corresponding to a state $\ket{\Psi}$ via the Choi--Jamio\l kowski isomorphism \cite{CiDuKrLe00} can be implemented using a system prepared in the state $\ket{\Psi}$ and local operations. It will be shown that the non--local content of the CPM, ${\cal E}_\Psi$, is characterized by the three parameters mentioned above and vice versa. This leads to the new approach of characterizing the entanglement of a multipartite state by the entangling capability of the operation which can be implemented using the state as the only non--local resource.

This suggests a new method to characterize the entanglement contained in an arbitrary multipartite state: First divide the Hilbert space into the entanglement classes resulting from the algorithm in \cite{Kr09}. Note that these classes can be easily characterized. Then, the entanglement of a state within a certain class should be qualified and quantified. Probably, the different classes might also lead to different
applications. For instance, for error correction, one way quantum
computing and quantum secret sharing, we have that all the employed
states have the property that all single qubit reduced states are completely
mixed. 

b) The other new insight into multipartite entanglement which we will derive here using the LU--equivalence criterion is the following. For any $n>2$ examples of $n$--qubit states, $\ket{\Psi}$ and
$\ket{\Phi}$ which have the properties that for any subsystem $A$, composed out of arbitrary many qubits, the eigenvalues of the reduced states, $\rho_A=\tr_{\neg A}(\proj{\Psi})$ and $\sigma_A=\tr_{\neg A}(\proj{\Phi})$
coincide, will be presented. Therefore, all the bipartite entanglement in those two states, measured with the von Neumann entropy of the reduced states, coincide. It will be shown, however, that the states are neither LU--equivalent nor
LOCC comparable. Therefore, neither $\ket{\Psi}$ can be mapped into $\ket{\Phi}$ by LOCC nor vice versa. Surprisingly, in those examples we will have $\ket{\Phi}=\ket{\Psi^\ast}$, where
$\ket{\Psi^\ast}$ denotes the complex conjugate of $\ket{\Psi}$ in
the computational basis. The fact that $\ket{\Psi}$ and $\ket{\Psi^\ast}$ can have so different non--local properties does not seem very physical.  As a consequence of the existence of these states it will be suggested to divide the Hilbert space into two subsets, in case $\ket{\Psi}$ is not LU--equivalent to $\ket{\Psi^\ast}$. One which corresponds to $\ket{\Psi}$ and one which corresponds to its complex conjugate. The non--local properties should then be investigated within one of the subsets since there will not be a physical measure which will distinguish between $\ket{\Psi}$ and $\ket{\Psi^\ast}$.

iii) Generalization of LU--equivalence criterion: It will be demonstrated how the solution of the LU--equivalence for pure $n$--qubit states can be generalized to mixed states and also to $d$--level systems.


\section{LU-equivalence of multipartite states}
\label{SecLU}

Here, we briefly summarize the necessary and sufficient conditions for the
existence of LU operations which transform two $n$--qubit states
into each other \cite{Kr09}. We will first review a standard form for multipartite states (see also \cite{CaHi00} and \cite{KrKr08}) and will provide some examples of states in their standard form. It has been shown that two generic multipartite states, i.e. states where none of the single qubit reduced states is proportional to the identity, are, like in the bipartite case, LU--equivalent iff their standard forms coincide \cite{Kr09}. For non--generic
states the systematic method to determine the local unitaries (if they exist) which interconvert two arbitrary states will be reviewed.

\subsection{Standard form of multipartite states}

Let us first recall the definition of the standard form for multipartite states. Any decomposition of a multipartite state which has the property that the single qubit reduced
states are all diagonal in the computational basis is called trace decomposition. It is obtained by applying local unitary transformations, $U_i^1$, which diagonalize the single qubit reduced states, $\rho_i$, i.e. $U^1_i\rho_i(U_i^1)^\dagger=D_i\equiv\mbox{diag}(\lambda_i^1,\lambda_i^2)$. A sorted trace decomposition, which we
will denote by $\ket{\Psi_{st}}$ in the following, is then defined as a trace decomposition with $\lambda_i^1\geq \lambda_i^2$. The sorted trace decomposition of a generic
state, $\ket{\Psi}$ with $\rho_i\neq \one$ $\forall i$ is unique up
to local phase gates. That is $U_1\ldots U_n \ket{\Psi_{st}}$ is a
sorted trace decomposition of a generic state, $\ket{\Psi}$, iff (up to a global
phase, $\alpha_0$) $U_i= Z_i(\alpha_i)\equiv \text{
diag}(1,e^{i\alpha_i })$. It is straightforward to impose certain conditions on the phases $\alpha_i$, $i\in\{0,\ldots, n\}$ in order to make the sorted trace
decomposition of generic states unique \cite{Kr09}. We call this unique sorted trace decomposition standard form of the multipartite state. Note that any state can be transformed into its standard form by local unitary operations.

Let us now also recall how the standard form can be defined for
states with $\rho_i=\one$, for some system $i$
\cite{KrKr08}. In this case the standard form can be chosen to be
$\lim_{\epsilon\rightarrow 0}\ket{\Psi(\epsilon)}$, where
$\ket{\Psi(\epsilon)}$ denotes the unique standard form of
$\sqrt{1-\epsilon}\ket{\Psi}+\sqrt{\epsilon} \ket{{\bf 0}}$, where the phase
gates are fixed by the same conditions as for generic states \cite{KrKr08}.
It should be noted here that for non--generic states the standard form
is not unique, as can be seen by the following simple example of a
three--qubit states. Both, the GHZ--state,
$\ket{\Psi}=\ket{000}+\ket{111}$ and $H H H \ket{\Psi}$ are
standard forms of the state $\ket{\Psi}$, however they do not
coincide.

Let us now derive the standard form for some examples. For
$2$--qubit states the standard form coincides with the Schmidt
decomposition \cite{NiCh00}. In \cite{Cohen} the standard form of
three--qubit states has been derived. In order to present the
standard forms of certain $n$--partite states we recall here the
notion of the so--called Locally Maximally Entangleable States
(LMESs) \cite{KrKr08}. LMESs have been introduced as a new,
physically motivated, classification of pure quantum states
describing $n$ qubits. A state is called LMES if local auxiliary
qubits can be attached to the system qubits in such a way that the
resulting state is maximally entangled in the bipartite splitting
system qubits versus auxiliary qubits. To be more precise, a state
$\ket{\Psi}$ is a LMES if there exist local control operations
$C_i=\proj{0}\otimes U_i^0+\proj{1}\otimes U_i^1$, with $U_i^{0,1}$
single qubit unitary operators acting on the system qubit $i$, such
that the $2n$--qubit state $C_1 \otimes C_2 \otimes \ldots \otimes
C_n \ket{\Psi}\ket{+}^{\otimes n}$, with
$\ket{+}=1/\sqrt{2}(\ket{0}+\ket{1})$, is a maximally entangled
state between the system and the auxiliary systems. This set of
states coincides with the set of states which can be used to encode
locally the maximum amount of $n$ independent bits. Prominent
examples of these states are the stabilizer states, which are used
for quantum error correction and one--way quantum computing. In
\cite{KrKr08} it has been shown that a state is LME iff it is
LU--equivalent to a state of the form \bea \label{LME}
\ket{\Psi}=\sqrt{\frac{1}{2^{n}}}\sum_{{\bf i}} e^{i \alpha_{{\bf
i}}}\ket{{\bf i}}\equiv U^\Psi_{ph}\ket{+}^{\otimes n},\eea where
$\alpha_{\bf{i}}\in \R$ and $U^\Psi_{ph}$ denotes the diagonal
unitary operator with entries $e^{i\alpha_{{\bf i}}}$. Thus, a state is LME iff there exists a product basis such that all
the coefficients of the state in this basis are phases.

Note that all those states can be easily transformed into their
trace decomposition by applying the local unitary operations
$HU_i$, where $U_i=\textrm{diag}(e^{i \phi_i},1)$, with $\cot(\phi_i)=
\frac{\langle X_i\rangle}{\langle Y_i\rangle}$ if $\langle Y_i\rangle \neq 0$ and $\phi_i=0$ else. To derive from the
trace decomposition the standard form one simply has to impose the
conditions on the local phase gates, as mentioned above.

\subsection{Criterion for LU--equivalence}

Since the standard form is unique for generic states we have, similarly to the bipartite case that two generic states
are LU--equivalent iff their standard forms are equivalent.

Let us now turn to the more complicated case of non--generic states. First, the condition of LU--equivalence for generic states is rewritten in the following way. It can be easily seen that the standard forms of two generic states, $\ket{\Psi}$ and $\ket{\Phi}$ are equivalent iff there exists a bit string ${\bf
k}=(k_1,\ldots, k_n)$, local phase gates $Z_i(\alpha_i)$, and a
global phase $\alpha_0$ s.t. \bea \label{LU} e^{i\alpha_0}
\bigotimes_i
 Z_i(\alpha_i)X_i^{k_i}W_i\ket{\Psi}=\bigotimes_i V_i
\ket{\Phi},\eea where $W_i$ ($V_i$) are local unitaries
which diagonalize $\rho_i$ ($\sigma_i$). That
is $\bigotimes_i W_i\ket{\Psi}$ and $ \bigotimes_i V_i
\ket{\Phi}$ are trace decompositions of $\ket{\Psi}$ and
$\ket{\Phi}$ resp.. For generic states, $k_i$ is chosen such that
the order of the eigenvalues of the single qubit reduced states of
$\bigotimes_i
 X_i^{k_i}W_i\ket{\Psi}$ and $\bigotimes_i V_i
\ket{\Phi}$ coincides and the phases $\alpha_i$ are chosen to fulfill the conditions mentioned above \cite{Kr09}. Note that the reason for the freedom of the phase gates in Eq. (\ref{LU}) is simply due to the fact that we have been considering only single qubit reduced states to define the trace decomposition of multipartite states.

Obviously, two arbitrary states $\ket{\Psi}$ and $\ket{\Phi}$ are LU--equivalent iff there exist local unitaries $V_i$ and $W_i$ and a bit string ${\bf
k}=(k_1,\ldots, k_n)$, and phases $\alpha_i$ s.t. Eq (\ref{LU}) is fulfilled. For non--generic states, a constructive method to determine the unitaries $V_i$ and $W_i$ in Eq. (\ref{LU}) has been presented in \cite{Kr09}. Once those unitaries are fixed it is then easy to decide whether or not there exist local phase gates for a certain bit string ${\bf k}$ such that Eq. (\ref{LU}) is fulfilled (see Lemma \ref{LemmaPhase} below).

Since we are going to determine the local unitaries which transform two states into each other in Sec \ref{secExamples} we review here the constructive method to compute the unitaries $V_k, W_k$.
First of all, it is easy to see that if $\ket{\Psi}$ is such that there exists some system $i$ such that  $\rho_i\neq \one$ the unitaries $V_i$ and $W_i$ can be determined by considering the necessary condition for LU--equivalence, $\rho_i=U_i\sigma_i U_i^\dagger$. Analogously to the generic case,
the equation
$D_i=\mbox{diag}(\lambda_1^i,\lambda_2^i)=W_i \rho_i
 W_i^\dagger=V_i \sigma_i V_i^\dagger$ determines $W_i$ and $V_i$ (and $k_i=0$)
uniquely up to a phase gate. Thus, for this case we have that
$\ket{\Psi}\simeq_{LU} \ket{\Phi}$ iff there exist two phases,
$\alpha_{i}$ and $\alpha_0$ and local unitaries $U_j$ such that
\bea \label{cond1} \phantom{,}_i\bra{l}W_i\Psi_s\rangle
=e^{i(\alpha_0+\alpha_i l)}\bigotimes_{j\neq
i}U_{j}\phantom{,}_i\bra{l}V_i\Phi_s\rangle \mbox{ for }
l\in\{0,1\}, \eea where $W_i,$ and $V_i$ are chosen
such that $D_i=\mbox{diag}(\lambda_1^i,\lambda_2^i)=W_i
\rho_i W_i^\dagger=V_i \sigma_i V_i^\dagger$.
Hence, if there is one system where the reduced state is not
completely mixed, then the problem of
LU--equivalence of $n$--qubit states can be reduced to the one of
$(n-1)$--qubit states. This statement can be easily generalized to
the case where more than one single qubit reduced state is not
completely mixed.


Let us now turn to the remaining case where $\rho_i=\one$
$\forall i$. Instead of considering the necessary conditions
$\rho_i=U_i\sigma_i U_i^\dagger$, one considers the necessary
conditions $\rho_{n_1,\dots,n_l,k}=U_{n_1}\ldots U_{n_l}
U_{k}\sigma_{n_1,\dots,n_l,k} U_{n_1}^\dagger \ldots
U_{n_l}^\dagger U_k^\dagger$, for some appropriately chosen set
$\{n_1,\ldots n_l,k\}$ and computes $U_k$ as a function of
$U_{n_1},\ldots U_{n_l}$. More precisely, it has been shown that if
$\ket{\Psi}=U_1\ldots U_n\ket{\Phi}$ and if there exist systems
$n_1, \ldots n_l$ and $k$ such that $\rho_{n_1, \ldots n_l,k}\neq
\rho_{n_1, \ldots n_l}\otimes \one_k$ then $V_k$ in Eq. (\ref{LU})
can be determined from the state $\ket{\Phi}$ and $W_k$ can be
determined as a function of the unitaries $U_{n_1},\ldots U_{n_l}$.
To be more specific, we assume without loss of generality that
$n_1=1, \ldots n_l=l$. Due to the condition
$\rho_{1,\ldots,l,k}\neq \rho_{1, \ldots l}\otimes \one_k$ it can
be shown that there exist at least two tuples $\i=(i_1,\ldots , i_l)$ and
$\j=(j_1,\ldots j_l)$ such that at least one of the hermitian
$2\times 2$ matrices
$B_{\i}^{\j}=A_{\i}^{\j}+(A_{\i}^{\j})^\dagger$ and $C_{\i}^\j=i
A_\i^\j-i(A_\i^\j)^\dagger$, where $A_{\i}^{\j}\equiv \tr_{\neg
k}[\ket{\i}\bra{\j}\ket{\Phi}\bra{\Phi}]$, is not proportional to
the identity. W. l. o. g. we assume that $\one \not\propto B_\i^\j
= \tr_{\neg k}[(\ket{\i}\bra{\j}+h.c)\ket{\Phi}\bra{\Phi}]$. Using
that $\ket{\Psi}=U_1\ldots U_n\ket{\Phi}$ we have \bea \label{Yi}
U_k B_\i^\j U_k^\dagger =\tr_{\neg k}[(\ket{{\bf i}}\bra{{\bf
j}}+h.c)\cdot U_1^\dagger \ldots U_l^\dagger
\ket{\Psi}\bra{\Psi}U_1 \ldots U_l].\eea Since $B_\i^\j$ is
hermitian we can diagonalize it as well as the right hand side of
Eq (\ref{Yi}). It can then be shown that $\ket{\Psi}=U_1\ldots
U_n\ket{\Phi}$ iff there exists $i_k\in\{0,1\}$, and $\alpha_0$ and
$\alpha_k$ such that
$e^{i\alpha_0}X^{i_k}U(\alpha_k)W_k(U_1,\ldots,U_l)
\ket{\Psi}=U_1\ldots V_k \ldots U_n\ket{\Phi}$, where $V_k$ is the
unitary which diagonalizes $B_\i^\j$ and can therefore be
determined directly from the state $\ket{\Phi}$ and
$W_k(U_1,\ldots,U_l)$ diagonalizes the right hand side of Eq.
(\ref{Yi}).

Note that this constructive method to compute $V_k,W_k$
is based on the necessary condition for
LU--equivalence given in Eq. (\ref{Yi}) for any $l$--tuples
$\bf{i},\bf{j}$. Since the $2\times 2$ matrices occurring in this
equation are hermitian, one can, similarly to the previous cases,
determine the unitaries $V_k, W_k$ by diagonalizing
these matrices. In contrast to before we will find here, that
$W_k$ might depend on $U_1, \ldots ,U_l$. Again, since those unitaries are obtained by diagonalizing a $2\times 2 $ matrix the phase gate occurring in Eq. (\ref{LU}) cannot be determined like that. This is the reason why the condition of LU--equivalence has been rewritten in the seemingly more complicated form presented in Eq. (\ref{LU}).


In order to check then whether or not there
exist phases $\alpha_i$ such that Eq. (\ref{LU}) is satisfied, the following lemma \cite{Kr09}, which will be proven here, has been used. We consider four $n$-- qubit systems which will be
denoted by $A,B,C,D$ respectively. The $i$-th qubit of system $A$
will be denoted by $A_i$, etc. Furthermore, we will use the
notation $\ket{\chi}_i=(\ket{0110}-\ket{1001})_{A_i,B_i,C_i,D_i}$
and $P^i_{AC}=\sum_{\bf k} \ket{{\bf k}}\bra{{\bf k}{\bf
k}}_{A_1,C_1,\ldots A_{i-1},C_{i-1},A_{i+1},C_{i+1}\ldots,
A_n,C_n}$. Similarly, we define $P^i_{BD}$ for systems $B,D$. For
a state $\ket{\Psi}$ we define $K_\Psi\equiv \{{\bf k} \mbox{ such
that }\bra{{\bf k}}\Psi\rangle=0\}$ and
$\ket{\Psi_{\{\bar{\alpha}_i\}}}=\ket{\Psi}+ 2e^{-i\bar{\alpha}_0}
\sum_{{\bf k}\in K_\Psi} e^{-i\sum_{i=1}^n \bar{\alpha}_i
k_i}\ket{{\bf k}}$ for some phases $\bar{\alpha}_i$ and
$\ket{\Psi_{\bf 0}}=\ket{\Psi}+2\sum_{{\bf k}\in K_\Psi} \ket{{\bf
k}}$.

\begin{lemma}
\label{LemmaPhase} Let $\ket{\Psi}, \ket{\Phi}$ be $n$--qubit
states. Then, there exist local phase gates, $Z_i(\alpha_i)$ and a
phase $\alpha_0$ such that
$\ket{\Psi}=e^{i\alpha_0}\bigotimes_{i=1}^n Z_i(\alpha_i)
\ket{\Phi}$ iff there exist phases $\{\bar{\alpha}_i\}_{i=0}^{n}$
such that (i) $|\bra {\bf i} \Psi_{\bf 0}\rangle|= |\bra {\bf i}
\Phi_{\{\bar{\alpha}_i\}}\rangle|$ $\forall {\bf i}$ and (ii)
$\bra{\chi}_i P^i_{AC} P^i_{BD} \ket{\Psi_{\bf 0}}_A\ket{\Psi_{\bf
0}}_B\ket{\Phi_{\{\bar{\alpha}_i\}}}_C\ket{\Phi_{\{\bar{\alpha}_i\}}}_D=0$
$\forall i\in \{1,\ldots ,n\}$.

\end{lemma}

Condition (ii) can be interpreted as follows. Taking two copies of the state $\ket{\Psi_{\bf
0}}$ and two copies of the state $\ket{\Phi_{\{\bar{\alpha}_i\}}}$
and projecting the four qubits, $A_i,B_i,C_i,D_i$ onto the state $\ket{\chi}_i$ leads to a $4(n-1)$--qubit state, which is in the kernel of $P^i_{AC} P^i_{BD}$ for any system $i$. Before proving Lemma \ref{LemmaPhase} we introduce here another lemma, which will be required for the proof. Using the same notation as before, we have

\begin{lemma}
\label{auxLemma}
$\ket{\Psi}$ can be converted into $\ket{\Phi}$ by local unitary
phase gates iff there exist phases $\{\bar{\alpha}_i\}_{i=0}^{n}$
such that $\ket{\Psi_{\bf 0}}$ is converted into
$\ket{\Phi_{\bar{\alpha}_i}}$ by local unitary phase
gates.\end{lemma}

\begin{proof}

If $\ket{\Psi}=e^{i\alpha_0}\bigotimes_{i=1}^n Z_i(\alpha_i)
\ket{\Phi}$, for some phases $\{\alpha_i\}$, then $K_\Psi=K_\Phi$ and choosing $\bar{\alpha}_i=\alpha_i$ for $i\in
\{0,\ldots, n\}$ fulfills the condition. To prove the inverse
direction we assume that there exist phases
$\{\bar{\alpha}_i\}_{i=0}^{n}$ such that $\ket{\Psi_{\bf
0}}=e^{i\alpha_0}\bigotimes_{i=1}^n Z_i(\alpha_i)
\ket{\Phi_{\bar{\alpha}_i}}$ for some phases $\{\alpha_i\}$. Due to the factor $2$ in the definition of $\ket{\Psi_{\bf 0}}$ and
$\ket{\Phi_{\bar{\alpha}_i}}$, this implies $K_\Psi=K_\Phi$.
Defining the projector $P=\sum_{{\bf k}\not\in K_\Psi} \proj{{\bf k}}$
we have $P\ket{\Psi_{\bf 0}}=\ket{\Psi}$ and
$Pe^{i\alpha_0}\bigotimes_{i=1}^n Z_i(\alpha_i)
\ket{\Phi_{\bar{\alpha}_i}}=e^{i\alpha_0}\bigotimes_{i=1}^n
Z_i(\alpha_i) P\ket{\Phi_{\bar{\alpha}_i}}$ and therefore
$\ket{\Psi}=e^{i\alpha_0}\bigotimes_{i=1}^n Z_i
(\alpha_i)\ket{\Phi}$.

\end{proof}
The reason for introducing this lemma is that it implies that if one wants to decide whether or not two states are up to local phase gates equivalent, one only needs to consider states where none of the coefficients in the computational basis vanish. Let us now use the lemma above to prove Lemma \ref{LemmaPhase}.
\begin{proof}

As mentioned above, due to the Lemma \ref{auxLemma} it remains to show that for any state
$\ket{\psi}$ with $\bra{{\bf k}}\psi\rangle\neq 0$ $\forall {\bf
k}$ we have that $\ket{\psi}=e^{i\alpha_0}\bigotimes_{i=1}^n Z_i
(\alpha_i) \ket{\phi}$, for some phases $\{\alpha_i\}$ iff condition (i) and (ii) in Lemma
\ref{LemmaPhase} are satisfied. Note that condition (ii) is
equivalent to $\bra{0{\bf k}}\psi\rangle \bra{1{\bf l}}\psi\rangle
\bra{1{\bf k}}\phi\rangle \bra{0{\bf l}}\phi\rangle= \bra{1{\bf k}}\psi\rangle
\bra{0{\bf l}}\psi\rangle \bra{0{\bf k}}\phi\rangle \bra{1{\bf l}}\phi\rangle, $
where $0,1$ is acting on system $i$ and ${\bf k},{\bf l}$ denote the
computational basis states of the remaining $n-1$ qubits.

Let us now prove the {\it only if} part: If
$\ket{\psi}=e^{i\alpha_0}\bigotimes_{i=1}^n Z(\alpha_i) \ket{\phi}$ for some phases $\{\alpha_i\}$ then
$\bra{{\bf i}}\psi\rangle=e^{i\phi_{{\bf i}}}\bra{{\bf
i}}\phi\rangle$, with $\phi_{{\bf i}}=\alpha_0+\sum_k \alpha_k
i_k$, which implies (i). Condition (ii) (for $i=1$) is then
equivalent to $e^{i(\phi_{0\k}+\phi_{1\l})}
x_{\k\l}=e^{i(\phi_{1\k}+\phi_{0\l})} x_{\k\l},$  where
$x_{\k\l}=\bra{0\k}\phi\rangle \bra{1\l}\phi\rangle \bra{1\k}\phi\rangle
\bra{0\l}\phi\rangle$. It is easy to see that this condition is
fulfilled since $e^{i(\phi_{0\k}-\phi_{1\k})}=e^{-i\alpha_1}$ $
\forall \k$. In the same way one can show that the conditions for
$i\neq 1$ are fulfilled.

To prove the {\it if} part, we first note that condition (i) implies that $\bra{{\bf i}}\Psi\rangle=e^{i
\phi_{\bf i}}\bra{{\bf i}}\Phi\rangle$, for some phases $\phi_{\bf
i}$. Condition (ii) (for $i=1$) implies then that
$e^{i(\phi_{0\k}-\phi_{1\k})} =e^{i(\phi_{0\l}-\phi_{1\l})}$ $\forall
\k,\l$,  since $x_{\k\l}=\bra{0\k}\phi\rangle \bra{1\l}\phi\rangle
\bra{1\k}\phi\rangle \bra{0\l}\phi\rangle \neq 0$ $\forall \k,\l$.
Thus, $e^{i(\phi_{0\k}-\phi_{1\k})}$ must be independent of $\k$ and
therefore, we have $e^{i(\phi_{0\k}-\phi_{1\k})}=e^{-i\alpha_1}$, for some phase $\alpha_1$. Equivalently, we have  $e^{i\phi_{k_1,\k}}=e^{i
(\alpha_1^{(k_1)}+\phi_{1\k})}$, where $\alpha_1^{(0)}=-\alpha_1$
and $\alpha_1^{(1)}=0$. Similarly, we obtain $e^{i(\phi_{k_10k_3,\ldots, k_n}-\phi_{k_11k_3\ldots,
k_n})}=e^{-i\alpha_2}$ and therefore
$e^{i\phi_{k_1,k_2,k_3\ldots,k_n}}=e^{i
(\alpha_1^{(k_1)}+\alpha_2^{(k_2)}+\phi_{11k_3,\ldots,k_n})}$.
Continuing in this way we find $e^{i\phi_{k_1,\ldots
k_n}}=e^{i\alpha_0}e^{i\sum_j \alpha_j k_j}$, where
$\alpha_0=\phi_{1\ldots 1}-\sum \alpha_i$. Thus, we have
$\ket{\psi}=e^{i\alpha_0}\bigotimes_{i=1}^n Z_i (\alpha_i)
\ket{\phi}$.
\end{proof}

It is important to note here that
the state on the right hand side of Eq. (\ref{LU}) is completely
determined using the method summarized above. Thus, the set $K_\Psi$ in Lemma \ref{LemmaPhase} can be
determined and therefore this lemma can be applied. The states are
LU--equivalent iff the conditions in Lemma \ref{LemmaPhase} are
fulfilled for some bit string ${\bf k}$. The unitaries which interconvert the states
are, up to the symmetry of the states, uniquely determined and are given
by $U_i=W^\dagger_i Z_i(\alpha_i) X^{k_i}V_i$ (up to a
global phase) \footnote{Note that the phases $\alpha_i$ can be
easily computed.}.

In summary, the LU--equivalence problem has been solved by presenting a systematic method to determine the local unitaries (if they exist) which interconvert the states. This has been achieved by determining $V_i,W_i$ in Eq. (\ref{LU}) by imposing necessary conditions of LU--equivalence, like $\rho_i=U_i\sigma_iU_i^\dagger$ and Eq. (\ref{Yi}). Once all the unitaries $V_i,W_i$ are determined (even as functions of some others), the states are LU--equivalent iff there exist local phase gates which interconvert the transformed states (after applying $\bigotimes_i V_i$, $\bigotimes_i W_i$ to $\ket{\Phi}$, $\ket{\Psi}$ respectively). This can then be easily decided by employing Lemma \ref{LemmaPhase}.

Before ending this section let us present here another way of checking whether or not two states are interconvertible by local phase gates. Due to Lemma \ref{auxLemma} we only need to consider states $\ket{\Psi}$, $\ket{\Phi}$ with $K_\Psi=K_\Phi={\O}$. Here and in the following we will denote by $\bigodot$ the Hadamard product, i.e. the component--wise product and by $/.$ we will denote the inverse operation, i.e. the component-wise division. For instance, if $\ket{\Psi}=\sum_{\i} a_\i \ket{\i}$ and $\ket{\Phi}=\sum_{\i} b_\i \ket{\i}$, with $b_\i\neq 0$ $\forall \i$, then $\ket{\Psi}/.\ket{\Phi}=\sum_{\i} a_\i/b_\i \ket{\i}$.

\begin{lemma} \label{LemmaPhase1} Let $\ket{\Psi}$ and $\ket{\Phi}$ be $n$--qubit states with $K_\Psi=K_\Phi={\O}$. Then, there exist phases $\{\alpha_i\}$ such that $\ket{\Psi}=e^{i\alpha_0}\bigotimes_i Z(\alpha_i) \ket{\Phi}$ iff i) $|\bra{\bf{i}}\Psi\rangle|=|\bra{\bf{i}}\Phi\rangle|$ and ii) $\ket{\Psi}/.\ket{\Phi}$ is a product state.  \end{lemma}

\begin{proof} (Only if): If $\ket{\Psi}=e^{i\alpha_0}\bigotimes_i Z(\alpha_i) \ket{\Phi}$ condition i) is obviously fulfilled. In order to show that condition ii) is fulfilled we use that  $e^{i\alpha_0}\bigotimes_i Z(\alpha_i) \ket{\Phi}=e^{i\alpha_0}\bigotimes_i Z(\alpha_i)\ket{+}^{\otimes n} \bigodot \ket{\Phi}$. Thus, $\ket{\Psi}/.\ket{\Phi}=e^{i\alpha_0}\bigotimes_i Z(\alpha_i)\ket{+}^{\otimes n}$, which is a product state.

(If): Due to condition i) we have that $\ket{\Psi}/.\ket{\Phi}=\sum e^{i\alpha_{{\bf i}}}\ket{{\bf i}}$, for some phases $\alpha_{\bf i}$. That is, $\ket{\Psi}/.\ket{\Phi}$ is a LME state. Due to condition (ii) this LME state must be a product state. i.e $\sum e^{i\alpha_{{\bf i}}}\ket{{\bf i}}=\bigotimes_i \ket{\phi_i}$, where $\ket{\phi_i}=e^{i\Phi_0^i}(\lambda_0^i\ket{0}+e^{i\Phi_1^i}\lambda_1^i \ket{1})$ with $\lambda_k^i\geq 0$. This implies that $e^{i\alpha_{{\bf i}}}=e^{i \Phi_0+\sum_{k} (\Phi_1^k)^{i_k}}$, where $\Phi_0=\sum_{k} \Phi_0^k$. Thus, the LME state is a product state iff it is equivalent to $e^{i\alpha_0}\bigotimes_i Z(\alpha_i)\ket{+}^{\otimes n}$ for some phases $\{\alpha_i\}$ and therefore $\ket{\Psi}=e^{i\alpha_0}\bigotimes_i Z(\alpha_i) \ket{\Phi}$.
\end{proof}

As mentioned above, Lemma \ref{auxLemma} can be used to generalize Lemma \ref{LemmaPhase1} to states, $\ket{\Psi}$ for which $K_\Psi \neq {\O}$. Note that condition ii) has a physical interpretation. The Hadamard product of two states $\ket{\psi}$, $\ket{\phi}$ corresponds to the state one would get by the following procedure. Let $\ket{\psi}$ ($\ket{\phi}$) describe the system $1_1,\ldots,n_1$ ($1_2,\ldots,n_2$) resp. and consider $n$ pairs of maximally entangled two--qubit states, $\ket{\Phi^+}=\sum_{i=0}^1 \ket{ii}$, describing systems $1_3,1_4, \ldots n_3,n_4$. Then, $\ket{\phi}\bigodot \ket{\psi}=\bigotimes_{i=1}^n \langle \Psi^0_{i_1,i_2,i_3} \ket{\psi}_{1_1,\ldots n_1}\ket{\phi}_{1_2,\ldots n_2} \bigotimes_{i=1}^n \ket{\Phi^+}_{i_3,i_4}$, where $\ket{\Psi^0}$ denotes the GHZ states here. This resembles the procedure of gate teleportation \cite{Ni01}.
Note that condition ii) is fulfilled iff there exists a product state, $\bigotimes_i\ket{\phi_i}$ such that $\ket{\Psi_{\bf 0}}=\ket{\Phi_{\bar{\alpha}_i}}\bigodot \bigotimes_i\ket{\phi_i}$.

\section{Additional Methods to compute the local unitaries}

\label{SecAdd}
We have seen before how the local unitaries which occur in Eq. (\ref{LU}) can be determined by imposing certain necessary conditions of LU--equivalence (see Eq. (\ref{Yi})).
One might also use other necessary conditions for LU--equivalence to determine those local unitaries. For instance, if $\ket{\Psi}=U_1 \ldots U_n\ket{\Phi}$ then $\tr_{1}(\rho_{1}\otimes \one_2 \rho_{12})=U_2\tr_{1}(\sigma_{1}\otimes \one_2 \sigma_{12})U_2^\dagger$ and $\tr_{23}(\rho_{123}\otimes \one_{1^\prime}\rho_{1^\prime 23}\otimes \one_1)=U_1 \otimes U_{1^\prime} \tr_{23}(\sigma_{123}\otimes \one_{1^\prime}\sigma_{1^\prime 23}\otimes \one_1)U_1^\dagger \otimes U_{1^\prime}^\dagger$. Of course, any generalization of these equations must be fulfilled too. Here we will use those and other necessary conditions for LU--equivalence to derive some additional methods to compute the unitaries in Eq. (\ref{LU}) for certain multipartite states. Depending on the properties of the states of interest one method or the other might be better suited. In Sec \ref{secExamples} we will use the various methods to compute the local unitaries directly, i.e. not as a function of other unitaries. This makes the characterization of LU---equivalence classes easier.

Here, we will first consider the LU--equivalence of two--qubit mixed states. Then we will focus on those states for which there exists at least one system $i$ with $\rho_i\neq \one$ and will derive a simple way to determine the unitaries in Eq. (\ref{LU}).

\subsection{Two--Qubit Mixed States}

\label{Sectwoquits}

For two--qubit mixed states, $\rho$, $\sigma$, necessary and sufficient conditions for LU--equivalence have been derived in \cite{Mak02}. However, if $\rho=\rho_{ij}$ ($\sigma=\sigma_{ij}$) denotes the reduced state of some systems $i,j$ of a multipartite state, $\ket{\Psi}$ ($\ket{\Phi}$) resp. and the aim is to investigate the LU--equivalence of $\ket{\Psi}$ and $\ket{\Phi}$, then one must determine all local unitaries, $U_i,U_j$, which fulfill $\rho_{ij}=U_i U_j\sigma_{ij} U_i^\dagger U_j^\dagger$ and then check if there exists one of them, which transforms the multipartite states into each other. We are going to show here how to achieve this task.

We have seen above that if there exists some system $i$ such that $\rho_i\neq \one$, then $V_i$, $W_i$ and $k_i$ in Eq. (\ref{LU}) can be determined by imposing the necessary condition $\rho_i=U_i\sigma_i U_i^\dagger$. Thus, it remains to consider the case where both reduced states are proportional to the identity which implies that $\rho=\one +\sum_{k,l} \Lambda_{k,l} \Sigma_k\otimes \Sigma_l$, where $\Lambda=\sum_{kl}\lambda_{kl}\ket{k}\bra{l}$ is real. Applying the local unitary operation
$U_1\otimes U_2$ to the state, $\rho$ leads to $U_1\otimes U_2 \rho U_1^\dagger \otimes U_2^\dagger= \one +\sum_{k,l} \Lambda^\prime_{k,l} \Sigma_k\otimes \Sigma_l$, with $\Lambda^\prime=O_1\Lambda O^T_2$.
Here, $O_1,O_2$ are real orthogonal matrices which are defined via the equation $U_i (\vec{n}\vec{\sigma})U_i^\dagger =(O_i\vec{n}
\vec{\sigma})$ for $i=1,2$. Using the singular value decomposition of the real matrix $\Lambda$, $\Lambda=O_1 D O_2^T$, where $D$ is a diagonal matrix, and $O_{1,2}$ are real and orthogonal, and the fact that the state $\one+\sum_{k,l} D_{k,k} \Sigma_k\otimes \Sigma_k$,
is Bell--diagonal shows that the eigenbasis of any two--qubit density matrix with completely mixed reduced states is
maximally entangled.

In order to show now under which conditions two two--qubit states are LU--equivalent we recall the following Lemma which was proven in \cite{KrCi01}.

\begin{lemma}\label{LemmaME} Any two maximally entangled basis of two qubits can be mapped into each other using local unitary operations \cite{KrCi01}. That is, if $\{\ket{\Psi_i}\}_{i=1}^4 $ and $\{\ket{\Phi_i}\}_{i=1}^4 $ denote two maximally entangled bases then there exist four phases $\gamma_i$, and local unitaries $U_1,U_2$ such that $\ket{\Psi_i}=e^{i\gamma_i}U_1 \otimes U_2 \ket{\Phi_i}$ $\forall i\in\{1,2,3,4\}$.
\end{lemma}

This lemma together with the fact that the eigenbasis of any two--qubit density matrix with completely mixed reduced states is
maximally entangled implies the following corollary.

\begin{corollary} \label{twoQubits} Let $\rho$, $\sigma$ be two--qubit density matrices with completely mixed reduced states.
Then $\rho\simeq_{LU} \sigma$ iff
$\mbox{eig}(\rho)=\mbox{eig}(\sigma)$.
\end{corollary}

Let us now consider two LU--equivalent states, $\rho,\sigma$ with $\rho_i=\sigma_i=\one$ for $i=1,2$ and derive some conditions on the local unitary operations,
which transform $\sigma$ into $\rho$. First we apply local unitaries, $W_i,V_i$ such that $\bar{\rho}=W_1W_2\rho W_1^\dagger W_2^\dagger=\one+ \sum_i (D_\rho)_i \Sigma_i \Sigma_i$ and $\bar{\sigma}=V_1V_2\sigma V_1^\dagger V_2^\dagger=\one+ \sum_i (D_\sigma)_i \Sigma_i \Sigma_i$, with $D_\rho=D_\sigma=\mbox{diag}(\lambda_1,\lambda_2,\lambda_3)$. We choose w.l.o.g. the order of $\lambda_i$ such that if there is no degeneracy $\lambda_1> \lambda_2> \lambda_3$, else $\lambda_1=\lambda_2$. If $D_\rho$ is not proportional to the identity it is easy to see that $\bar{\rho}=\bar{U}_1 \bar{U}_2 \bar{\sigma} \bar{U}^\dagger_1 \bar{U}^\dagger_2$ implies that $\bar{U}_i$ is of the form $Z(\alpha_i) X^{k_i}$, for some phase $\alpha_i$ and $k_i\in \{0,1\}$ \footnote{This can be easily seen by noting that the condition $O_1D_\rho O_2^T=D_\rho$ implies that $O_iD^2O_i^T=D^2$, for $i=1,2$, which defines $O_i$ uniquely up to $R_z(\alpha)R_x(\pi)^k$, for $k=0,1$.}. Thus, if $\rho$ and $\sigma$ denote for instance the reduced state of system $1$ and $2$ of some multipartite state, $\ket{\Psi}$, $\ket{\Phi}$ respectively, then $\ket{\Psi}\simeq_{LU} \ket{\Phi}$ iff Eq. (\ref{LU}) is fulfilled for $V_1,V_2$ and $W_1,W_2$ such that
$W_1W_2\rho W_1^\dagger W_2^\dagger=\bar{\sigma}=V_1V_2\sigma V_1^\dagger V_2^\dagger=\one+ \sum_i (D_\rho)_i \Sigma_i \Sigma_i$ where $D_\rho=\mbox{diag}(\lambda_1,\lambda_2,\lambda_3)$ with $\lambda_i$ sorted as mentioned above.

Otherwise, if $D_\rho$ is proportional to the identity and $\rho\neq\one$, we apply the local unitaries $V_i$ and $W_i$ defined above and denote the resulting states again by $\rho,\sigma$ respectively. In this case we find $\rho=\sigma=\one-\lambda \proj{\Psi^-}$ for some $\lambda\neq 0$. Then, any pair of unitaries $U_1$, $U_2$ which transforms $\sigma$ into $\rho$ must fulfill that $U_2=U_1$. Hence, in this case we have that if $\rho$ and $\sigma$ denote for instance the reduced state of system $1$ and $2$ of some multipartite state, $\ket{\Psi}$, $\ket{\Phi}$ respectively, then $\ket{\Psi}\simeq_{LU} \ket{\Phi}$ iff Eq. (\ref{LU}) is fulfilled for $V_1=V_2=\one$, $k_1=k_2=0$, $\alpha_1=\alpha_2$, and $W_1=W_2=e^{i\beta_1 X_i} e^{i\gamma_1 Z_i}$, for some phases $\beta_1,\gamma_1$. Note that if $\rho_{12}=\one-\lambda \proj{\Psi^-}$ with
$\lambda\neq 1$ then $\ket{\Psi}\simeq_{LU}\ket{\Phi}$
implies that $\phantom{.}_{12}\langle \Psi^- \ket{\Psi}\simeq_{LU} \phantom{.}_{12}\langle \Psi^-
\ket{\Phi}$ since $U \otimes U\ket{\Psi^-}=\ket{\Psi^-}$ for any unitary $U$. Thus, similar to
the case where $\rho_i\neq \one$ one would simply measure those
systems where the reduced state is a full rank Werner state \cite{Wer89}.

In the remaining case, that is if $\rho=\one$, it is clear that considering the two--qubit reduced state will not help us to find any condition on the local unitaries.

So far we have seen that whenever there exist two systems $i,j$ such that $\rho_{ij}$
is not LU--equivalent to $\one+\lambda\sum_i\Sigma_i\otimes
\Sigma_i$, then the unitaries $V_i,V_j,W_i,W_j$ and $k_i,k_j$ in Eq. (\ref{LU}) can be
easily determined. We are going to show next that in this case also
other unitaries, $V_l$ and $W_l$, for $l\not\in \{i,j\}$ can be
easily computed. As before we consider the case where $\rho_i=\one$
$\forall i$. In the following lemma we will say that the unitaries can be determined by considering a certain operator, if they can be determined using the fact that the operator for the state $\ket{\Psi}$ and the one for the state $\ket{\Phi}$ must be LU--equivalent if $\ket{\Psi}\simeq_{LU} \ket{\Phi}$.

\begin{lemma}
\label{Lemma3}
If there exist systems $i,j$ such that $\rho_{ij}$ is not
LU--equivalent to $\one+\lambda\sum_i\Sigma_i\otimes \Sigma_i$, for
any $\lambda \in \R$, then, for any system $l$ for which either
$\rho_{il}\neq \one$ or $\rho_{jl}\neq \one$, $V_l, W_l$ and $k_l$
can be determined by either considering $\rho_{il}$ or $\rho_{jl}$
or by considering $\tr_i(\rho_{ij}\rho_{il})$ or
$\tr_j(\rho_{ij}\rho_{jl})$.
\end{lemma}

\begin{proof}
If $\rho_{il}$ or $\rho_{jl}$ is not LU--equivalent to $\one+\lambda\sum_i\Sigma_i\otimes \Sigma_i$, for some $\lambda$, then $V_l, W_l$ and $k_l$ can be determined as shown above. Otherwise, we assume without loss of generality that $\rho_{il}\neq \one$. Then we have that $\rho_{il}=\one+\sum_{i_1,i_2=1}^3
\Lambda_{i_1 i_2}\Sigma_{i_1}\otimes \Sigma_{i_2}$ with
$\Lambda$ proportional to a real orthogonal matrix and $\rho_{ij}=\one+\sum_l
\tilde{\Lambda}_{l_1 l_2} \Sigma_{l_1}\otimes \Sigma_{l_2}$ with $\tilde{\Lambda}$ not proportional to a real orthogonal matrix. Then, we find
$\tr_i(\rho_{ij}\rho_{il})=\one +
\sum_{i_1,i_2}\tilde{\Lambda}_{i_1 l_2}\Lambda_{i_1 i_2}\Sigma_{l_2}
\Sigma_{i_2}$, where the matrix $(\tilde{\Lambda})^T\Lambda$
is not orthogonal. Since the unitaries $W_j,V_j$ are already fixed, the equation $W_j W_l \tr_i(\rho_{ij}\rho_{il}) W_j^\dagger W_l^\dagger=V_j V_l \tr_i(\sigma_{ij}\sigma_{il}) V_j^\dagger V_l^\dagger$ 
determines $V_l,W_l$ and $k_l=0$.
\end{proof}

\subsection{States where there exists a system $i$ with $\rho_i\neq \one$}

Let us now turn to the case where there exists at least one system $i$ such that its reduced state is not completely mixed. Without loss of generality we chose $i=1$. Let us assume that the states of interest, $\ket{\Psi}$ and
$\ket{\Phi}$ do have sorted trace decomposition. That is, in
particular the unitaries $V_1,W_1$ which make $\rho_1$ and
$\sigma_1$ diagonal in the computational basis (see Eq. (\ref{LU}))
have already been applied. Then we know that $U_1=Z(\alpha_1)$ for some phase $\alpha_1$. We will present now several methods to determine the unitaries $V_i,W_i$ of Eq (\ref{LU}). As in \cite{Kr09} the idea is to construct out of the states $\ket{\Psi}$ and $\ket{\Phi}$ non--degenerate $2\times 2$ matrices which must be LU--equivalent in case $\ket{\Psi}$ and $\ket{\Phi}$ are, e.g. $U\rho_i U^\dagger=\sigma_i$, or $\tr_1[(\rho_1\otimes \one_i)\rho_{1i}]=
U_i \tr_1[(\sigma_1\otimes \one_i)\sigma_{1i}] U_i^\dagger$. Those necessary conditions of LU--equivalence can then be used to fix $W_i, V_i$ and $k_i$ in Eq. (\ref{LU}). The local phase gates, which cannot be fixed in this way, must be determined at the end using one of the lemmata in Sec. \ref{SecLU}.

Since the states $\ket{\Psi}$ and $\ket{\Phi}$ have sorted trace decomposition, we have $\ket{\Psi}=\sqrt{p_1} \ket{0}\ket{\Psi_0}+ \sqrt{1-p_1}
\ket{1}\ket{\Psi_1},$ with
$\bra{\Psi_i}\Psi_j\rangle=\delta_{i,j}$ and $ p_1> 1/2$ and
$\ket{\Phi}=\sqrt{p_2} \ket{0}\ket{\Phi_0}+ \sqrt{1-p_2} \ket{1}\ket{\Phi_0}$ with
$\bra{\Phi_i}\Phi_j\rangle=\delta_{i,j}$ and $p_2> 1/2$, which is just the Schmidt decomposition for the bipartite splitting, system $1$ versus the rest. Then we have that the states are LU--equivalent iff 1) $p_1=p_2$ and 2)
there exist phases, $\phi$, $\alpha_1$ and unitaries $U_j$ such
that \bea \label{oneNot}\ket{\Psi_0}&=& e^{i\gamma_1}
\bigotimes_{j\neq 1} U_j \ket{\Phi_0}\\ \nonumber \ket{\Psi_1}&=&
e^{i\gamma_2} \bigotimes_{j\neq 1} U_j \ket{\Phi_1},\eea

where $\gamma_1=\phi+\alpha_1$ and $\gamma_2=\phi-\alpha_1$ (see Eq. (\ref{cond1})). Note that the last two conditions are fulfilled iff
$\beta_1\proj{\Psi_0}+\beta_2\proj{\Psi_1}=\bigotimes_{j\neq 1} U_j
(\beta_1\proj{\Phi_0}+\beta_2\proj{\Phi_1})\bigotimes_{j\neq 1} U_j^\dagger$
for all values of $\beta_1,\beta_2$.

There are several ways now to compute the unitaries $V_i,W_i$ in
Eq. (\ref{LU}). First of all, if $\rho_i\neq \one$ then
$U_i=Z(\alpha_i)$ for some phase $\alpha_i$. Let us now consider
the case where $\rho_i=\one$. If $\rho_{1i}\neq \rho_1\otimes \one$
Eq. (\ref{Yi}) can be used to compute $V_i,W_i$, and $k_i$. In
particular, we would consider one of the matrices $ B_{l,m}\equiv
\tr_{\neq i}(\ket{\Psi_l}\bra{\Psi_m}+ h.c.)$ or $C_{l,m}\equiv
\tr_{\neq i}(i\ket{\Psi_l}\bra{\Psi_m}+ h.c.)$ for some $l,m$ and
diagonalize this matrix in order to compute $W_i, V_i$ and $k_i$.

It is the aim of this section to derive some other methods to
determine those unitaries. First, we will use the fact that Eq.
(\ref{Yi}) must be fulfilled for any values of $l,m$ if the states
are LU--equivalent. Considering certain combinations of those
equations will lead to other approaches to determine the unitaries
$V_i,W_i$ and the bit value $k_i$. In the second part of this
section we will show how a combination of the Eqs. (\ref{oneNot})
can be used to compute those unitaries.

If $\rho_{1i}\neq \rho_1\otimes \one$ then $V_i$ and $W_i$ can be easily computed as follows. First of all, it is clear that if
$\rho_{1}\neq \sigma_{1}$ then they states are not LU--equivalent. Thus, we assume that $\rho_{1}= \sigma_{1}$. The fact that $\rho_1\neq\one$ and $\rho_i=\one$ implies that $\rho_{1i}= \one+a\z\one+\sum_{j_1,j_2=1}^3
\Lambda_{j_1j_2}\Sigma_{j_1}\otimes \Sigma_{j_2}$, for some $a\in \R$ and where $\Lambda\neq 0$ since
$\rho_{1i}\neq \rho_1\otimes \one$. Similarly we have $\sigma_{1i}= \one+a\z\one+\sum_{j_1,j_2=1}^3
\Gamma_{j_1j_2}\Sigma_{j_1}\otimes \Sigma_{j_2}$. As mentioned before, the two states are LU--equivalent, i.e.  $\rho_{1i}=U_1U_i \sigma_{1i} U^\dagger_1 U^\dagger_i$ iff there exists a real orthogonal matrix $O_i$ and a phase $\alpha_1$ such that
$\Lambda=O_1\Gamma O_i^T$, where $O_1=R_z(\alpha_1)$.
As explained in Sec. II A, if $\Lambda$ is not orthogonal, then the unitaries $V_i$, $W_i$ can be easily computed. Otherwise, we use the following
necessary condition for LU--equivalence: $(\rho_1 \otimes
\one_i)\rho_{1i}=U_1 U_i[(\sigma_{1}\otimes \one_i)\sigma_{1i}]
U^\dagger_1 U^\dagger_i$ and therefore $\tr_1[(\rho_1 \otimes
\one_i)\rho_{1i}]=U_i\tr_1[(\sigma_1 \otimes
\one_i)\sigma_{1i}]U_i^\dagger$. Since
$\tr_1[(\rho_1 \otimes
\one_i)\rho_{1i}]=\one+a \sum_j \lambda_{3j}\Sigma_j$ the
equation above can be used to determine $V_i$ and $W_i$ (and
$k_i=0$) as those operators which diagonalize
$\tr_1[(\rho_1 \otimes
\one_i)\rho_{1i}]$ and $\tr_1[(\sigma_1 \otimes
\one_i)\sigma_{1i}]$
respectively. That is, unless $\Lambda_{3j}=0$ $\forall j$ $V_i$ and
$W_i$ are defined by the equation  $V_i\tr_1[(\rho_1 \otimes
\one_i)\rho_{1i}]
V_i^\dagger=W_i\tr_1[(\sigma_1 \otimes
\one_i)\sigma_{1i}]W_i^\dagger=\mbox{diag}(\gamma_1,\gamma_2)$,
for some $\gamma_i$. If $\Lambda_{3i}=0$ $\forall i$ $\Lambda$ cannot
be orthogonal and therefore the unitaries can be determined as
explained in Sec. \ref{Sectwoquits}. Thus, if $\rho_{1i}=U_1U_i
\sigma_{1i} U^\dagger_1 U^\dagger_i$ the methods described above
will lead to the unitaries $V_i$ and $W_i$ in Eq. (\ref{LU}) unless
$\rho_{1i}= \rho_1\otimes \one$.


Another method to compute the unitaries is the following. Instead
of considering the single equation $\ket{\Psi}=U_1\ldots
U_n\ket{\Phi}$ we use the fact that the basis for the first system
has been fixed. Therefore, we can use both equations given in Eq.
(\ref{oneNot}). Note that the states $\ket{\Psi_0}$, $\ket{\Psi_1}$
and the states $\ket{\Phi_0}$, $\ket{\Phi_1}$ are orthogonal
respectively. In \cite{WaSh00} it has been shown that two orthogonal
multipartite pure states can be perfectly distinguished using local
operations. We are going to use this result now in order to
determine the unitaries $V_i,W_i$. It is easy to see that for any
pair of orthogonal states, $\ket{\Psi_0}$, $\ket{\Psi_1}$, there
exist local unitaries, $W_i$ such that \cite{WaSh00} \bea M_i\equiv
\tr_{\neg i}(\ket{\Psi_0}\bra{\Psi_1})=W_i N_i W_i^\dagger,\eea
where $N_i$ is a off--diagonal matrix with $(N_i)_{1,2}=a_i,
(N_i)_{2,1}=b_i$, for some complex numbers $a_i,b_i$. If $|a_i|\neq
|b_i|$, i.e. if $N_i N_i^\dagger\not\propto \one$, it is easy to
see that by imposing the condition that $\vert a_i\vert > \vert
b_i\vert$ this equation determines $W_i$ uniquely (up to a phase
gate). If $\vert a_i\vert=\vert b_i\vert \neq 0$, $M_i$ is, up to
a global phase, a hermitian traceless matrix. Thus, in this case we
would chose $W_i$ such that $M_i=e^{i \bar{\alpha}} W_i^\dagger D_i
W_i$, for some phase $\bar{\alpha}$ and $D_i$ diagonal. Defining
$V_i$ in the same way for the state $\ket{\Phi}$ we have that
$\ket{\Psi}\simeq_{LU} \ket{\Phi}$ iff Eq. (\ref{LU}) has a
solution for the so chosen matrices $W_i,V_i$ \footnote{
Note that if $\rho_{1i}=\rho_1\otimes \one$ the condition that
$\tr_{\neg i}(\ket{\Psi_0}\bra{\Psi_1})\propto
\bra{0}_1\rho_{1i}\ket{1}_1$ is local unitarily equivalent to
$\tr_{\neg i}(\ket{\Phi_0}\bra{\Phi_1})$ cannot, like any other
necessary condition of LU--equivalence considering just
$\rho_{1i}$, enable us to determine the unitaries $V_i$ and $W_i$.
In fact, in this case we find $M_i=0$.}.

Before concluding this section we would like to mention another method
to compute the unitaries in Eq. (\ref{oneNot}) for a general state
with $\rho_1\neq \one$. It is based on the following observation.
Let us denote by $\ket{\Psi_0}$ ($\ket{\Psi_1}$) a (unnormalized)
state describing systems $2,\ldots,n$ ($2^\prime,\ldots,n^\prime$)
respectively. Then $\bra{\Psi^-}_{ii^\prime}\Psi_0,\Psi_1\rangle=0$
iff the state $\ket{\Psi}=\ket{0}\ket{\Psi_0}+\ket{1}\ket{\Psi_1}$
is either a product state in the bipartite splitting system $1$
versus the rest, or system $i$ versus the rest. This can be easily
verified as follows. We write
$\ket{\Psi_k}=\ket{0}_i\ket{\Psi_{k0}}+\ket{1}_i\ket{\Psi_{k1}}$,
for $k\in \{0,1\}$. Then
$\bra{\Psi^-}_{ii^\prime}\Psi_0,\Psi_1\rangle=0$ iff either a)
$\ket{\Psi_{00}}\ket{\Psi_{11}}=0$ and
$\ket{\Psi_{01}}\ket{\Psi_{10}}=0$ which implies that either
$\ket{\Psi}=\ket{k}_1\ket{\Psi_k}$ or
$\ket{\Psi}=\ket{k}_i(\ket{0}_1\ket{\Psi_{0k}}+\ket{1}_1\ket{\Psi_{1k}})$ for $k\in \{0,1\}$;
or b) $\ket{\Psi_{00}}=a\ket{\Psi_{01}}$ and
$\ket{\Psi_{10}}=a\ket{\Psi_{11}}$, for some $a$. In this case we
find $\ket{\Psi}=(a\ket{0}+\ket{1})_i\otimes
(\ket{0}\ket{\Psi_{01}}+\ket{1}\ket{\Psi_{11}})$. Note that if
$\ket{\Psi}$ is a product state in any bipartite splitting, system
$i$ versus the rest, then the unitaries $V_i,W_i$ and the bit value
$k_i$ in Eq. (\ref{LU}) can obviously be easily determined. If $\ket{\Psi}$
is not a product state in this splitting then we can combine the
two equations in Eq. (\ref{oneNot}) to

\bea \label{PsimProj}
\bra{\Psi^-}_{ii^\prime}\Psi_0,\Psi_1\rangle=\\ \nonumber =e^{i(\gamma_1+\gamma_2)}\bigotimes_{j\neq
i,1} U_j \bigotimes_{j^\prime\neq i^\prime,1^\prime} U_{j^\prime}
\bra{\Psi^-}_{ii^\prime}\Phi_0,\Phi_1\rangle,\eea where we used
that $\ket{\Psi^-}=U\otimes U\ket{\Psi^-}$ for any unitary $U$.
This approach will be useful if there are only a few unitaries not
determined. For instance, if $\ket{\Psi}$ is a three--qubit state.
Choosing w.l.o.g. $i=2$ we have
$\bra{\Psi^-}_{22^\prime}\Psi_0,\Psi_1\rangle=e^{i(\gamma_1+\gamma_2)}
U_3 U_{3^\prime} \bra{\Psi^-}_{22^\prime}\Phi_0,\Phi_1\rangle$,
which can then be used to determine $U_3$, or equivalently $V_3,W_3$ and $k_3$. Of course, the
projection onto the singlet state can also be performed on more
systems.

In summary, in this subsection we have explained some simple ways
to compute $V_i, W_i$ and $k_i$ for states which have the
properties that $\rho_i=\one$ and that there exists some
system $j$ such that $\rho_j\neq  \one$. For states with
$\rho_{ji}\neq \rho_j\otimes \one$, the unitaries can be easily
computed using either that $(\rho_1\otimes \one)
\rho_{1i}\simeq_{LU}(\sigma_1\otimes \one) \sigma_{1i}$ is a
necessary condition for LU--equivalence or that the states
$\ket{\Psi_0}$ and $\ket{\Psi_1}$ in Eq. (\ref{oneNot}) are orthogonal. For general
states (not requiring that $\rho_{ji}\neq \rho_j\otimes \one$) where there exists a system $j$ with $\rho_j\neq \one$ Eq.
(\ref{PsimProj}) (and its generalizations) can be used to find new
conditions on the unitaries. Note that if $\rho_1\neq  \one$ and $\rho_{12}= \rho_1\otimes \one$ then we have \bea \label{rho12}
\ket{\Psi}=\sqrt{p}\ket{0}(\ket{0}\ket{\Psi_{00}}+\ket{1}\ket{\Psi_{01}})
\\ \nonumber
+\sqrt{1-p}\ket{1}(\ket{0}\ket{\Psi_{10}}+\ket{1}\ket{\Psi_{11}}),\eea
with $\bra{\Psi_{ij}}\Psi_{kl}\rangle
=1/2\delta_{ik}\delta_{jl}$, where $p$ and $1-p$ denote the eigenvalues of $\rho_1$.


\section{Examples}

\label{secExamples}

We are going to employ now the algorithm presented in \cite{Kr09} and the results shown in the previous section to characterize the LU--equivalence classes of up to five qubits. We will show that in all these cases it is not necessary to determine some unitaries as functions of some others, but that it is always possible to determine them directly.

\subsection{Two--qubit states}
\label{subsectwoQubits}

The standard form of a two--qubit state is $
\ket{\Psi}=\lambda_1\ket{00}+\lambda_2\ket{11}$, with $\lambda_1\geq \lambda_2\geq 0$, which coincides
with the Schmidt decomposition \cite{NiCh00}. It is a well--known
fact that bipartite states are LU--equivalent iff their Schmidt
coefficients coincide. Let us now demonstrate how this result can
be rederived with the method presented in \cite{Kr09} for two
qubits. If $\lambda_1\neq \lambda_2$, i.e. $\rho_i\neq \one$ then,
$\ket{\Psi}\simeq_{LU}\ket{\Phi}$ iff $eig(\rho_1)=eig(\sigma_1)$,
i.e. iff the Schmidt coefficients $\lambda_i$ are the same. For
$\lambda_1=\lambda_2$ we have that $\rho_1=\rho_2=\one$ and
therefore the states are LU--equivalent iff $eig(\rho)=eig(\sigma)$
(Lemma \ref{twoQubits}), which is obviously the case.

\subsection{ Three--qubit states}

First we transform both states, $\ket{\Psi}$ and $\ket{\Phi}$ into their sorted
trace decomposition. If there exists some $i$ such that
$\rho_i\neq\one$, we know that $U_i=Z(\alpha_i)$. Without loss of
generality we assume $i=1$. Then we have that the states
$\ket{\Psi}=\sqrt{p_1}\ket{0}\ket{\Psi_0}+\sqrt{1-p_1}\ket{1}\ket{\Psi_1}$,
for some $p_1>1/2$ and $\langle \Psi_i \ket{\Psi_j}=\delta_{ij}$ and
$\ket{\Phi}=\sqrt{p_2}\ket{0}\ket{\Phi_0}+\sqrt{1-p_2}\ket{1}\ket{\Phi_1}$,
for some $p_2>1/2$ and $\langle \Phi_i \ket{\Phi_j}=\delta_{ij}$, are LU--equivalent iff 1) ${\mbox
eig}(\rho_1)={\mbox eig}(\sigma_1)$, i.e. iff $p_1=p_2$,
and 2) there exist local unitaries $U_2,U_3$ and two phases, $\phi$, $\alpha_1$ such that \bea \label{twoStates}\ket{\Psi_0}&=&e^{i\gamma_1} U_2\otimes U_3 \ket{\Phi_0} \\
\nonumber \ket{\Psi_1}&=&e^{i\gamma_2} U_2\otimes U_3
\ket{\Phi_1},\eea

where $\gamma_1=\phi+\alpha_1$ and $\gamma_2=\phi-\alpha_1$. Since the states are not LU--equivalent if $p_1\neq p_2$ we assume that $p_1=p_2$ and show now in detail how the unitaries can be computed. According to the method summarized in Sec. II we distinguish the two cases 1) $\rho_{12}\neq \rho_1\otimes \one$ and 2) $\rho_{12}= \rho_1\otimes \one$. Since the rank of $\rho_{12}$ cannot be larger than two, the second case is only possible if $p=1$, i.e. the states $\ket{\Psi}$ is a product state. Then, the two states are LU--equivalent iff the two--qubit states $\phantom{.}_1\langle 0\ket{\Psi}$ and $\phantom{.}_1\langle 0\ket{\Phi}$ are (see Sec \ref{subsectwoQubits}). In the first case we have that either 1a) at least one of the two states $\ket{\Psi_i}$ is not maximally entangled or 2b) both are maximally entangled. To investigate the case 1a) we assume without loss of generality that $\ket{\Psi_0}$ is not maximally entangled and denote by $W_i$ ($V_i$) the local unitaries which map $\ket{\Psi_0}$ ($\ket{\Phi_0})$) into its standard form respectively, i.e. $\ket{\Psi_0}=W_1^\dagger W_2^\dagger(\sqrt{q_1}\ket{00}+\sqrt{1-q_1}\ket{11})$, with $q_1 > 1/2$ ($\ket{\Phi_0}=V_1^\dagger V_2^\dagger(\sqrt{q_2}\ket{00}+\sqrt{1-q_2}\ket{11}$ with $q_2>1/2$). Obviously, $\ket{\Psi_0}\simeq_{LU}\ket{\Phi_0}$ iff $q_1=q_2$. In this case, the most general unitaries which transform $\ket{\Phi_0}$ into $\ket{\Psi_0}$, i.e. $\ket{\Psi_0}=e^{i\gamma_1} U_2\otimes U_3 \ket{\Phi_0}$ are of the form  $U_i=W_i Z(\alpha_i) V_i^\dagger $, for some phase $\alpha_i$. Thus, we have that $\ket{\Psi}\simeq_{LU}\ket{\Phi}$ iff there exists phases $\alpha_i$ such that Eq. (\ref{LU}) is fulfilled for $k_1=k_2=k_3=0$, $V_1=W_1=\one$, and $V_i,W_i$ as defined above. This condition can then be easily checked using Lemma \ref{LemmaPhase}. In case 1b) we have that both, $\ket{\Psi_0}$ and $\ket{\Psi_1}$ are maximally entangled and are therefore LU--equivalent to $\ket{\Phi^\pm}$. Thus, $\ket{\Psi}$ is LU--equivalent to $\ket{\Phi}$ in this case iff $\phantom{.}_1\langle k \ket{\Phi}$ is maximally entangled for $k=0,1$. Since any state in this class is LU--equivalent to the state $\sqrt{p_1}\ket{0}\ket{\Phi^+}+\sqrt{1-p_1}\ket{1}\ket{\Phi^-}$, the unitaries which map two states within this class into each other can be easily computed.

Let us now consider the remaining case where all single qubit reduced states are completely mixed. Since $\rho_1=\rho_2= \one$ the eigenbasis of $\rho_{12}$ is
maximally entangled and therefore LU--equivalent to the
Bell--basis (see Sec \ref{Sectwoquits}). Thus, any state with $\rho_i=  \one$ $\forall i$ is
LU--equivalent to
$\ket{\Psi}=\ket{\Phi^+}\ket{0}+\ket{\Phi^-}\ket{1}=\one \otimes \one \otimes H
\ket{\Psi_0}$, where $\ket{\Psi_0}=1/\sqrt{2}(\ket{000}+\ket{111})$ denotes the GHZ--state.

In summary we obtained the following necessary and sufficient condition for LU--equivalence: The two three--qubit states $\ket{\Psi}$ and $\ket{\Phi}$ are LU--equivalent iff one of the following conditions are fulfilled:

\bi
\item[1)] $E_i(\ket{\Psi})=E_i(\ket{\Phi})=1$ $\forall i$ (i.e. $\rho_i=\one \forall i$).
    \item[2)] There exists some system $i$ such that $E_i(\ket{\Psi})=E_i(\ket{\Phi})=0$ and $E_j(\ket{\Psi})=E_j(\ket{\Phi})$ for some system $j\neq i$.
\item[3)] There exists some system $i$ such that $0<E_i(\ket{\Psi})<1$ and $E_i(\ket{\Psi})=E_i(\ket{\Phi})$
    and either \bi \item[3a)] $E_j(_i\langle k\ket{\Psi})=E_j(_i\langle k\ket{\Phi})=1$ for $k=0,1$ for some system $j\neq i$ holds; or
\item[3b)] $E_j(_i\langle k\ket{\Psi})=E_j(_i\langle k\ket{\Phi})<1$ for one value of $k\in \{0,1\}$ and for the unitaries which can be easily and directly determined in this case there exists a bit string ${\bf k}$ and local phase gates such that Eq. (\ref{LU}) has a solution.
    \ei
\ei

For three qubits the polynomial invariants which define the different LU--equivalence classes are known \cite{GrRo98}. In \cite{StCaKr10} we will compare them to the criterion derived here and investigate the measures of entanglement which are required to identify the different classes.

This completes the solution to the LU--equivalence problem of three--qubit states. However, in order to illustrate the method presented in \cite{Kr09}, we will apply it to the most complicated case, where $\rho_i=\one $ $\forall i$. We will show now that all these states are LU--equivalent without using the fact that they are LU--equivalent to the GHZ--state, $\ket{\Psi_0}$. In other words, we will determine now the unitaries $\{U_i\}$ such
that $\ket{\Psi}=\ket{\Psi_0}=U_1U_2U_3\ket{\Phi}$, where
$\ket{\Phi}=S_1^\dagger S_2^\dagger S_3^\dagger \ket{\Psi}$. Here, $S_i$ are some fixed unitaries. Since the rank of
$\rho_{12}$ is two, we can compute $U_2$ as a function of $U_1$. We
find $U_2^\dagger\tr_{\neg 2}(\proj{i}_1
\proj{\Psi})U_2=U_2^\dagger\proj{i}_2 U_2=S_2^\dagger
\tr_{/2}(W_1^\dagger \proj{i}_1 W_1\proj{\Psi})S_2$, where
$W_1=U_1S_1^\dagger$. Since the rank of these matrices is one,
we have (due to the fact that $W_1$ is unitary) that either
$\bra{0}W_1^\dagger\ket{0}=0$ (which implies that
$\bra{1}W_1^\dagger\ket{1}=0$) or $\bra{1}W_1^\dagger\ket{0}=0$
(which implies that $\bra{0}W_1^\dagger\ket{1}=0$). Thus, $W_1=X^{k_1}Z(\alpha_1)$ (up to a global phase), for some $k_i\in \{0,1\}$ and some phase $\alpha_1$. Due to the symmetry of the state the same holds
true for all $W_i=U_iS_i^\dagger$. Thus, we have that
$\ket{\Psi}=U_1U_2U_3\ket{\Phi}$ iff there exists $k_1,k_2,k_3$ and
phases $\alpha_i$ such that $\ket{\Psi}=e^{i\alpha_0}\bigotimes_i Z(\alpha_i) X^{k_i} \ket{\Psi}$. Of course it is straightforward
to determine the remaining parameters, but in order to continue
with the algorithm we use Lemma \ref{LemmaPhase} to show for which values of $k_i$ the states are
up to phase gates local unitary equivalent. The first condition, 
$|\bra{\i}\Psi\rangle|=|\bra{\i}\Phi\rangle|$ $\forall \i$, implies that
$k_1=k_2=k_3$. Then we have $\ket{\Psi}$ is LU--equivalent to
$\ket{\Phi}$ iff there exist phases $\alpha_i$ such that
$\ket{\Psi}=e^{i\alpha_0}\bigotimes_i Z(\alpha_i) \ket{\Psi}$,
which is true iff $\alpha_0=0$ and $e^{i(\alpha_1+\alpha_2+\alpha_3)}=1$. Thus,
using the method above we found $U_i=Z(\alpha_i)X^{k_i}S_i$, with
$e^{i(\alpha_1+\alpha_2+\alpha_3)}=1$ and $k_1=k_2=k_3$.
The unitaries are not uniquely defined due to the symmetry of the
state.

\subsection{Four--qubit states}

In this subsection we will consider the LU--equivalence of two four--qubit states, $\ket{\Psi}$ and $\ket{\Phi}$. Similarly to the other
cases we transform both states, $\ket{\Psi}$ and $\ket{\Phi}$ into
their sorted trace decomposition. Like before the solution can of
course be found using the method summarized in Sec. II. However, we
will show here that also in this case it is possible to determine
the unitaries $V_i$ and $W_i$ (and the bit values $k_i$) directly.
That is, it will not be necessary to consider some of the unitaries
as variables and to determine those unitaries by solving the
equations which occur in Lemma \ref{LemmaPhase}.

According to the general method, we first distinguish the cases 1)
there exists some system $i$ with $\rho_i\neq \one$ and 2) $\rho_i=
\one$ for any system $i$. In the first case we choose $i=1$ and
know that $U_1=Z(\alpha_1)$ for some phase $\alpha_1$. Then we can
either have that 1a) $\rho_{12}\neq \rho_1\otimes \one$ or 1b)
$\rho_{12}= \rho_1\otimes\one$. In case 1a) $W_2,V_2$, and $k_2$ in
Eq. (\ref{LU}) can easily be determined using the methods presented
in Sec. \ref{SecAdd}. Then system $1$ and $2$ can be measured in
the computational basis leading to four two--qubit states. The
remaining unitary operators can then be easily found. If
$\rho_{12}= \rho_1\otimes\one$ (case 1b) we will show next that at least one of the unitaries $V_i,W_i$ and the bit--values $k_i$ can be determined for $i\in\{3,4\}$. First note that the eigenvalues of $\rho_{34}$ are
$p/2,p/2,(1-p)/2,(1-p)/2$ with $p\neq 1/2$ and therefore $\rho_{34}$ is neither $\one$ nor  LU--equivalent to $\one-\lambda
\proj{\Psi^-}$, for any value of $\lambda$. Thus, the unitaries
$W_i,V_i$ and $k_i$ for $i=3,4$ can be easily computed unless $\rho_{34}=\rho_3 \otimes \one$ (or $\rho_{34}=\one \otimes \rho_4$). However, in this case, since $\rho_{34}\neq \one$, the unitaries $V_3,W_3$ and the bit--value $k_3$ (or $V_4,W_4$ and the bit--value $k_4$ resp.) can be easily determined.

Let us now consider the remaining case where $\rho_i=\one$ $\forall
i$ (case 2). There, all two--qubit reduced states are LU--equivalent
to $\one+\sum_i \lambda_i \Sigma_i\otimes \Sigma_i$ and the reduced
states are LU--equivalent to each other iff the eigenvalues are the
same (Lemma \ref{twoQubits}). We distinguish now the two cases 2a) $\rho_{12}\neq \one$ and
2b) $\rho_{12}=\one$. In the first case $U_2$ can be determined as
a function of $U_1$ (see Sec \ref{SecAdd}). Let us apply local
unitaries to both states, $\ket{\Psi}$ and $\ket{\Phi}$ such that
$\rho_{12}, \rho_{34}$ and $\sigma_{12},\sigma_{34}$ are both
Bell--diagonal (see Lemma \ref{LemmaME}). We sort the eigenvalues in such a way that if there is three--fold degeneracy, then the states are such that $\rho_{12}=\rho_{34}=\one-\lambda \proj{\Psi^-}$ \footnote{Note that if $\rho_{12}=\one-\lambda \proj{\Psi^-}$ then $\rho_{34}=\one-\lambda \proj{\Psi^-}$ follows from the fact that all single qubit reduced states are completely mixed and that the state, describing all four systems is pure.}. The resulting states will
again be denoted by $\ket{\Psi}$, $\ket{\Phi}$ respectively. If $ \rho_{12}\neq \one-\lambda \proj{\Psi^-}$, then the unitaries can be easily determined (see Sec. \ref{SecAdd}). In the
"worst" case, where $\rho_{12}=\one-\lambda \proj{\Psi^-}$ we only
find $U_2=U_1$. We will show next that also in this case $U_2$, or
more precisely $V_2,W_2$ and $k_2$ in Eq. (\ref{LU}) can be
directly computed. That is we will not need to compute any of those
unitaries as a function of some others.

Using Lemma \ref{LemmaME} it is easy to see that any state
$\ket{\Psi}$ with $\rho_{12}=\one-\lambda \proj{\Psi^-}$, is
LU--equivalent (up to a global phase) to a state
$\ket{\Phi^+}\ket{\Phi^+}+e^{i\gamma_1}\ket{\Phi^-}\ket{\Phi^-}+
e^{i\gamma_2}\ket{\Psi^+}\ket{\Psi^+}+\sqrt{1-\lambda}e^{i\gamma_3}\ket{\Psi^-}\ket{\Psi^-}$,
for some phases $\gamma_i$. Since the operations $\Sigma_i\otimes \Sigma_i$ for $i\in \{1,2,3\}$ always change the sign of two states out of the four Bell states, we can choose $\gamma_1,\gamma_2 \leq \pi$. We are going to show next that two
states of this form with the choice $\gamma_1,\gamma_2 \leq \pi$ are LU--equivalent iff the complex coefficients
which occur here coincide.

Let us denote by $U_{mb}$ the $4\times 4$ unitary matrix, which
transforms the computational basis into the magic basis,
i.e. $U_{mb}\ket{00}=\ket{\Phi^+},U_{mb}\ket{01}=-i\ket{\Phi^-},
U_{mb}\ket{10}=\ket{\Psi^-},U_{mb}\ket{11}=-i \ket{\Psi^+}$. It is
a well--known fact that for any $U_1,U_2$ unitary we have that
$U_{mb}^\dagger U_1\otimes U_2 U_{mb}= O$, where $O$ is a real and
orthogonal $4\times 4$ matrix. Furthermore, it is easy to see that $O_i\equiv U_{mb}^\dagger U_i\otimes U_i
U_{mb}$ can be written as $O_i= \tilde{O}_i \oplus \proj{01}$
($U_{mb}\ket{01}=\ket{\Psi^-}$), where $\tilde{O}_i$ is a
three--dimensional rotation.

Since $\rho_{12}=\sigma_{12}=\one-\lambda \proj{\Psi^-}$ and therefore $ \rho_{34}=\sigma_{34}=\one-\lambda \proj{\Psi^-}$ we know that there exist local unitaries, $U_i$ which map $\ket{\Phi}$ into $\ket{\Psi}$ iff $U_2=U_1$ and $U_4=U_3$. Thus, we have $\ket{\Psi}\simeq_{LU}\ket{\Phi}$ iff there exist real orthogonal matrices, $O_i= \tilde{O}_i \oplus \proj{01}$, with $\tilde{O}_i$ is a
three--dimensional rotation such that
\bea \label{Psitilde}\ket{\tilde{\Psi}}\equiv
U_{mb}^\dagger\otimes U_{mb}^\dagger \ket{\Psi}=O_1\otimes O_3
\ket{\tilde{\Phi}}.\eea

Note that $\ket{\tilde{\Psi}}=\ket{00}\ket{00}-e^{i\gamma_1}\ket{10}\ket{10}-
e^{i\gamma_2}\ket{11}\ket{11}+\sqrt{1-\lambda}e^{i\gamma_3}\ket{01}\ket{01}$, and that the phases $\gamma_i$ are not local phases. Similarly we have $\ket{\tilde{\Phi}}\equiv U_{mb}^\dagger\otimes U_{mb}^\dagger \ket{\Phi} =e^{i\bar{\gamma}_0}(\ket{00}\ket{00}+e^{i\bar{\gamma}_1}\ket{10}\ket{10}+
e^{i\bar{\gamma}_2}\ket{11}\ket{11}+\sqrt{1-\bar{\lambda}}e^{i\bar{\gamma}_3}\ket{01}\ket{01})$, for some phases $\bar{\gamma}_i$ and coefficient $\bar{\lambda}$. Using now that the real and orthogonal matrices $O_1, O_3$ are of the form $\tilde{O}_i \oplus \proj{01}$ it is easy to see that $\ket{\Psi}\simeq_{LU}\ket{\Phi}$, i.e. Eq. (\ref{Psitilde}) is satisfied iff $\{e^{i\gamma_i}\}_{i=1}^2=
\{e^{i\bar{\gamma}_i}\}_{i=1}^2$ and $\sqrt{1-\lambda}e^{i\gamma_3}=\sqrt{1-\bar{\lambda}}e^{i\bar{\gamma_3}}$.

For the case 2b), where $\rho_{12}=\one$, which implies that $\rho_{34}= \one$ we
write $\ket{\Psi}=\sum_{ij}\ket{i,j}\ket{\psi_{i,j}}$ with
$\bra{\psi_{ij}}\psi_{kl}\rangle=\delta_{ik}\delta_{jl}$. Since
these states form an ON--basis we can find a $4\times 4$ unitary $U$
such that $\ket{\psi_{i,j}}=U\ket{ij}$. Recall that any two--qubit
unitary operator $U$ can be written as $U=U_1 \otimes U_2 U_d V_1 \otimes V_2$,
where $U_d$, the non--local content of $U$, is diagonal in the magic basis, i.e. $U_d=e^{i(\phi_1 X\otimes X+\phi_2 Y\otimes Y+\phi_3 Z\otimes Z)}$, for some phases $\phi_i$ \cite{KrCi01}. Note that $U_d$ can be
made unique by imposing certain conditions on the phases, $\phi_i$ \cite{HaCi00}. We transform the
state by local unitary operations into the form $\one_{12}
\otimes U_d \sum_{ij}\ket{ij} \ket{ij}= \one_{12} \otimes U_d
\ket{\Phi^+}_{13} \ket{\Phi^+}_{24}$. Then the two states, $\ket{\Psi}=\one_{12} \otimes U_d (\Psi)
\sum_{ij}\ket{\Phi^+}_{13} \ket{\Phi^+}_{24}$ and $\ket{\Phi}=\one_{12} \otimes U_d (\Phi)
\sum_{ij}\ket{\Phi^+}_{13} \ket{\Phi^+}_{24}$
are LU--equivalent iff $U_d(\Psi)=U_d(\Phi)$. Thus, these LU--equivalence classes are characterized by $E_{ij}(\ket{\Psi})=2$ for some systems $i,j$ and the three parameters, $\phi_1,\phi_2,\phi_3$, which define the non--local content of $U_d$. In Sec \ref{SecEntanglement} we will give a physical meaning to these parameters and discuss its generalization.

Using the fact that any four--qubit state with $\rho_{12}=\one$ is LU--equivalent to the state $\ket{\Psi}=\one_{12} \otimes
U_d \sum_{ij}\ket{ij} \ket{ij}$, where
$U_d=U_{mb}\mbox{diag}(1,e^{i\phi_1},e^{i\phi_2},e^{i\phi_3})U_{mb}^\dagger$, for some phases $\phi_i$, it is also easy to rederive the result that there
exists no four qubit state with the property that all two--qubit
reduced states are completely mixed, i.e. $\rho_{ij}=\one$ $\forall
i,j$. This can be seen as follows. Since $\ket{\Psi}=\one_{12} \otimes
U_d \sum_{ij}\ket{ij} \ket{ij}$, we find  $\rho_{13}=\tr_{4}(U_d\proj{\Phi^+}_{13} \otimes \one_4
U_d^\dagger)$. Then the conditions $\rho_{13}=\rho_{23}=\one$ imply that
$\cos(\phi_i)=0$ and $\cos(\phi_i-\phi_j)=0$ $\forall i,j$. Since
it is impossible to fulfill those equations simultaneously, we have
that there exists no four--qubit state such that $\rho_{ij}=\one$ $\forall
i,j$. This implies that the case 2b) is actually contained in 2a).
In fact, it corresponds to the case where $\lambda=0$.

In summary, for the four qubit case we have the following possibilities:
\bi \item [1a )] There exist some systems $i$ and $j$ with
$\rho_i\neq \one$ and $\rho_{ij}\neq \rho_i\otimes \one$: Without
loss of generality we choose $i=1$ and $j=2$. Then $W_1=V_1=\one$
and $k_1=0$ in Eq. (\ref{LU}) and $V_2,W_2$ and $k_2$ can be easily
computed using the methods presented in Sec. \ref{SecAdd}.
\item[1b)]  There exists some system $i$ with $\rho_i\neq \one$ and
some system $j$ with $\rho_{ij}= \rho_i\otimes\one$: Without loss
of generality we chose $i=1$ and $j=2$. Then the unitaries
$W_i,V_i$ and $k_i$ for $i=3,4$ can be easily computed by
considering $\rho_{34}$ which can neither be $\one$ nor
$\one-\lambda \proj{\Psi^-}$, for any value of $\lambda$. \ei

In both cases at least for two systems the operators $V_i,W_i$ and
the bit values $k_i$ can be determined. Thus, measuring those two
systems in the computational basis leads to four equations for two--qubit states. The missing
operators $V_i,W_i$ and the bit values $k_i$ can then be easily
computed. The states are LU--equivalent iff there exist some phases
$\alpha_i$ such that Eq. (\ref{LU}) has a solution, which can easily be
checked using Lemma \ref{LemmaPhase}. \bi \item[2a)] For any system $i$, $\rho_i=\one$ and there exists a system $j$ such that $\rho_{ij}\neq
\one$, for some system $i$: Without loss of generality we chose
$i=1$ and $j=2$. First we apply local unitaries to both states,
$\ket{\Psi}$ and $\ket{\Phi}$ such that $\rho_{12}, \rho_{34}$ and
$\sigma_{12},\sigma_{34}$ are all Bell--diagonal (see Sec  \ref{Sectwoquits}).
The resulting states will again be denoted by $\ket{\Psi}$,
$\ket{\Phi}$ respectively. If $\rho_{12}\neq\one-\lambda
\proj{\Psi^-}$ for some $\lambda\in \R$ (which implies that also
that $\rho_{34}\neq\one-\lambda \proj{\Psi^-}$ for some $\lambda\in
\R$), the unitaries $V_i,W_i$ and $k_i$ in Eq. (\ref{LU}) can be
directly computed using the methods of Sec. \ref{SecAdd}. If
$\rho_{12}=\one-\lambda \proj{\Psi^-}$ for some $\lambda\in \R$ we
map $\ket{\Psi}$ into the form
$\ket{\Phi^+}\ket{\Phi^+}+e^{i\gamma_1}\ket{\Phi^-}\ket{\Phi^-}+
    e^{i\gamma_2}\ket{\Psi^+}\ket{\Psi^+}+\sqrt{\lambda}e^{i\gamma_3}\ket{\Psi^-}\ket{\Psi^-}$, for some phases $\gamma_i$ with $\gamma_{1,2}<\pi$. Two states of this form are LU--equivalent iff their complex coefficients which occur here coincide.
\item[2b)]  For any system $i$, $\rho_i=\one$ and there exists a system $j$
such that $\rho_{ij}= \one$, for some system $i$: Without loss of
generality we chose $i=1$ and $j=2$. In this case the state is
LU--equivalent to the state $\ket{\Psi}=\one_{12} \otimes U_d
(\Psi) \sum_{ij}\ket{ij} \ket{ij}$, where
$U_d(\Psi)=U_{mb}\mbox{diag}(1,e^{i\phi_1},e^{i\phi_2},e^{i\phi_3})U_{mb}^\dagger$
can be chosen uniquely. Then two states are LU--equivalent iff
$U_d(\Psi)=U_d(\Phi)$. Note that since there exists no four--qubit state with
$\rho_{ij}=\one$ $\forall i,j$ case 2b) is contained in 2a).
 \ei

Similarly to the three--qubit case, we will consider now the most
complicated case (for the algorithm proposed in \cite{Kr09}) and
show how the unitaries which transform two LU--equivalent states
into each other can be determined. The most complicated case for
four qubits is the one where $\rho_{ij}=\one$, for some systems
$i,j$ which we chose to be $1,2$. Thus, we consider the example
$\ket{\Psi}=\one_{12}\otimes
U_{34}\ket{\Phi^+}_{13}\ket{\Phi^+}_{24}$, where we choose
$U_{34}=U_{mb}\mbox{diag}(1,e^{i\phi},e^{i\phi},1) U_{mb}^\dagger$ such that $\rho_{12}=\rho_{34}=\rho_{23}=\one$.
It can be easily shown that $\rho_{13}=\rho_{24}=1/4(\one+\x\otimes \x+\cos(\phi)(\z\otimes
\z-Y\otimes Y)$. Our aim is to determine $U_i$ such that
$\ket{\Psi}=U_1\ldots U_4 \ket{\Phi}$, where
$\ket{\Phi}=S_1^\dagger \ldots S_4^\dagger \ket{\Psi}$, for some
given unitaries $S_i$. Since $\rho_{13}\neq \one$ and
$\rho_{24}\neq \one$ we can compute $U_3$ ($U_4$) as a function of $U_1$ $(U_2)$ respectively.
Considering Eqs. (\ref{Yi}) for all values of $l$ and $m$ simultaneously we have
$\rho_{13}= 1/4(\one+\x\otimes \x+\cos(\phi)(\z\otimes \z-Y\otimes
Y))=U_1U_3(S_1^\dagger S_1^\dagger\rho_{13}S_1 S_3 )U_1^\dagger
U_3^\dagger$. It is straight forward to see that the last equation can only be fulfilled if $U_1 S_1^\dagger=U_3
S_3^\dagger=\Sigma_{k}$, for $k\in \{0,1,2,3\}$, where $\Sigma_{0}=\one$. Similarly we find
$U_2 S_2^\dagger=U_4 S_4^\dagger=\Sigma_{l}$. Thus, we have  
$U_i=\Sigma_{k_i} S_i$, where $k_1=k_3=k$ and $k_2=k_4=l$. It is
straightforward to show that $\ket{\Psi}=\Sigma_k\Sigma_l
\Sigma_k\Sigma_l \ket{\Psi}$ for certain values of $k,l$ (e.g.
$k=0,l=3$ or $k=1,l\in\{0,1\}$). Again, the reason, why the
unitaries are not uniquely determined using this method is because
of the symmetries of the state.


\subsection{Five--qubit states}

Instead of considering now, similarly to the other cases, all
possible classes of five--qubit states we consider here one of the
hardest examples to illustrate the method presented in \cite{Kr09}.
First, we are going to construct a five--qubit state, $\ket{\Psi}$,
which has the property that all the two--qubit reduced states are
completely mixed. Then, we consider the two states $\ket{\Psi}$ and
$\ket{\Phi}=S_1\otimes\ldots \otimes S_5\ket{\Psi}$ for some local
unitaries $S_i$ and compute the unitaries, $U_i$, which map
$\ket{\Phi}$ into $\ket{\Psi}$, using the algorithm presented in
\cite{Kr09} and summarized in Sec \ref{SecLU}. We will show that
also in this case it will not be necessary to determine any of the
unitaries as a function of some others, but that it will be
possible to determine them directly.

In order to construct the $5$--qubit state, $\ket{\Psi}$ with $\rho_{ij}=\one$ $\forall i,j$ we write $\ket{\Psi}=\ket{0}\ket{\Psi_0}+\ket{1}\ket{\Psi_1}$, where $\langle \Psi_i \ket{\Psi_j}=\delta_{ij}$.
As shown above, any four--qubit
state which has the property that $\rho_{12}=\one$
is LU--equivalent to a state $\one_{12}\otimes U_d \sum_{i,j}
\ket{ij}\ket{ij}$, where $U_d=U_{mb}\mbox{diag}(1,e^{i\alpha_1},e^{i\alpha_2},e^{i\alpha_3})
U_{mb}^\dagger$. Imposing now also that $\rho_{23}=\rho_{14}=\one$ and that the phases $\alpha_i$ fulfill $0\leq \alpha_i< \pi$, we find $\ket{\Psi_0}=\one_{12}\otimes U_{d_1} \sum_{i,j}
\ket{ij}\ket{ij}$, where $U_{d_1}=U_{mb}\mbox{diag}(1,e^{i\alpha_1},e^{i\alpha_1},1)
U_{mb}^\dagger$. It is easy to see that two states, $\ket{\Psi_0},\ket{\Psi_1}$, of this form are orthogonal to each other iff
$\alpha_2=\alpha_1+\pi$. We consider now the $5$ qubit state
$\ket{\Psi}=\ket{0}\ket{\Psi_0}+\ket{1}\ket{\Psi_1}$ with
$\ket{\Psi_0}=\one_{12}\otimes U_{d_1} \sum_{i,j} \ket{ij}\ket{ij}$
and $\ket{\Psi_1}=\one_{12}\otimes U_{d_2} \sum_{i,j}
\ket{ij}\ket{ij}$, with
$U_{d_2}=U_{mb}\mbox{diag}(1,-e^{i\alpha_1},-e^{i\alpha_1},1)
U_{mb}^\dagger$. It is straight forward to show that all two--qubit states, $\rho_{ij}$ are completely mixed. Note that $\ket{\Psi}=\ket{+}(\ket{\Phi^+,\Phi^+}+\ket{\Psi^+,\Psi^+})+e^{i\alpha}\ket{-}(\ket{\Phi^-,\Phi^-}+\ket{\Psi^-,\Psi^-})$.

Let us now consider the state $\ket{\Phi}=S_1\otimes\ldots \otimes
S_5\ket{\Psi}$ for some local unitaries $S_i$ and compute the
unitaries, $U_i$, which map $\ket{\Phi}$ into $\ket{\Psi}$. Since
$\rho_{ij}=\one$ we choose $U_1,U_2$ as parameters. Note that
$\rho_{123}\neq \one$, since it can have at most rank $4$. Thus, we
can compute the unitary $U_3$ as a function of $U_1$ and $U_2$
considering $\rho_{123}$, which can be easily shown to be
$1/8(\one+X\otimes X\otimes X$). We have $U_3^\dagger
\ket{\Psi}=U_1 U_2 \one U_4 U_5\ket{\Phi}$ and therefore
$U_3^\dagger \tr_{\neg 3}(\ket{kl}_{12}\bra{ij}
\proj{\Psi})U_3=\tr_{\neg 3}(U_1^\dagger U_2^\dagger
\ket{kl}_{12}\bra{ij}U_1 U_2 \proj{\Phi})$ for any $i,j,k,l$. It can be easily shown
that $\tr_{\neg 3}(\ket{00}_{12}\bra{11} \proj{\Psi})=X/2$ and that
$\tr_{\neg 3}(U_1^\dagger U_2^\dagger \ket{00}_{12}\bra{11}U_1 U_2
\proj{\Phi})=x S_3 X S_3^\dagger$, where $x$ depends on $U_1,U_2$.
Since only $x$ depends on $U_1$ and $U_2$, $U_3$ can be directly computed
(not only as a function of $U_1,U_2$). We find $ U_3=e^{i\alpha_3 \x}S_3^\dagger$ for some phase $\alpha_3$. Thus, denoting by
$\ket{\tilde{\Psi}}=H_3\ket{\Psi}$  we have
$\ket{\Psi}\simeq_{LU}\ket{\Phi}$ iff there exist local unitaries
$U_1,U_2,U_4,U_5$ and a phase $\alpha_3$ such that
$\ket{\tilde{\Psi}}=U_1U_2 Z(\alpha_3) H S_3^\dagger
U_4U_5\ket{\Phi}$, where $S_3$ is determined. Projecting now the
third system onto $\ket{0}$ we find a state of system 1245 with the
property that $\rho_{24}\propto \one+X\otimes X$ . Imposing then
the necessary condition $\rho_{24}=U_2U_4\sigma_{24}U_2^\dagger
U_4^\dagger$ of LU--equivalence leads immediately to
$U_2=e^{i\alpha_2 X} S_2^\dagger$ and $U_4=e^{i\alpha_4 X}
S_4^\dagger$. Similarly we find the other unitaries. Thus, we have
$\ket{\Psi}\simeq_{LU}\ket{\Phi}$ iff there exist phases $\alpha_i$
such that $\bigotimes_{i=1}^5 H_i\ket{\Psi}=e^{i\alpha_0}
\bigotimes_i Z(\alpha_i)(\bigotimes_{i=1}^5 H_i
S_i^\dagger)\ket{\Phi}$. The existence of the phases can be easily
verify either by looking at the coefficients in the computational
basis, or by employing Lemma \ref{LemmaPhase}.

\section{LOCC incomparability}
\label{SecLOCC}
The results presented here lead also to conditions for the existence of
more general operations transforming one state into the other,
namely Local Operations and Classical Communication (LOCC). This is
due to the fact that two multipartite states, having the same
marginal one--party entropies, are either LU--equivalent, or
LOCC--incomparable \cite{BePo01,Kempe}, i.e. none of the states can be mapped into the other by LOCC.

In this section we will show that for any $n>2$ there exists a pair of $n$--qubit states, $\{\ket{\Psi},\ket{\Phi}\}$ such that for any bipartition $A/B$, where $A$ contains $a$ qubits and $B$ $n-a$ qubits, $eig(\rho_A)=eig(\sigma_A)$, but the states are not LU--equivalent. In particular, the entropies of the reduced states of any  subsystem coincide, i.e. all bipartite
entanglement, measured with the von Neumann entropy of the reduced
states is the same for both states. Since the eigenvalues of all single qubit reduced states coincide those states are not even LOCC comparable \cite{BePo01,Kempe}.  Surprisingly, in those examples $\ket{\Phi}$ will be the complex conjugation (in the computational basis) of $\ket{\Psi}$. That is, for any $n$ ($n>2$) there exist $n$--qubit states which are not even LOCC comparable to its complex conjugate. The consequence of the existence of these states will be discussed in Sec \ref{SecEntanglement}.

Note that, for $3$--qubit states examples of such states have already been presented in \cite{AcAn00}. There, the fact that the states are not LU--equivalent has been proven by employing a polynomial invariant of degree $12$. Here, we will use the necessary and sufficient condition for LU--equivalence presented in \cite{Kr09} and Sec. \ref{SecLU} to prove that the considered $n$--qubit states are not LU--equivalent. First we will present a three--qubit state, $\ket{\Psi}$ which is not LU--equivalent and therefore not even LOCC comparable to its complex conjugate. Then we will generalize this example to $n$--qubit states.

We consider the LME states
$\ket{\Psi}=U_{123}(\phi)\ket{+}^{\otimes 3}$ and
$\ket{\Phi}=U_{123}(\phi+\pi)\ket{+}^{\otimes 3}$, where
$U_{123}(\alpha)$ is the three--qubit phase gate defined by
$U_{123}=\one-(1-e^{i\phi})\proj{111}$. As mentioned in Sec \ref{SecLU} an arbitrary LMES, $\ket{\Psi}$, can be easily transformed into its trace
decomposition, $\ket{\Psi_{tr}}$, by applying the local unitary operations
$HZ(\phi_i)$, where $\phi_i$ is chosen such that
$\cot(\phi_i)=\frac{\langle X_i\rangle}{\langle Y_i\rangle}$. For
the symmetric state $\ket{\Psi}$ we find $\langle X_i\rangle=1/4 (3
+ \cos(\phi))$ and $\langle Y_i\rangle= \sin(\phi)/4$ and therefore
$\cot(\phi_i)=\cot(\phi) + 3 \csc(\phi)$ for $i=1,2,3$. For
$\phi=\pi/2$ the marginal entropies of $\ket{\Psi}$ and
$\ket{\Phi}$, which is equivalent to the complex conjugate of
$\ket{\Psi}$ in this case, coincide. However, it is easy to show
that $\ket{\Psi_{tr}}/.\ket{\Psi_{tr}}^\ast$ is not a product state and therefore
the states $\ket{\Psi}$ and $\ket{\Psi}^\ast$ are not
LU--equivalent (see Lemma \ref{LemmaPhase1}). Moreover, due to the fact that the eigenvalues of all the reduced states are the same for $\ket{\Psi}$ and $\ket{\Psi^\ast}$, those two states are not even
LOCC comparable. Note that those two states have the same bipartite
entanglement (considering any bipartite splitting) and the same value for the tangle \cite{CoKu00}, the value of which is the same for a state and its complex conjugate.

Let us now generalize this example to $n$--qubit states
(for $n>2$). That is the two $n$--qubit states,
$\ket{\Psi}=U_{1,\ldots,n}(\pi/2)\ket{+}^{\otimes n}$ and
$\ket{\Phi}=\ket{\Psi^\ast}$ have the property that
$eig(\rho_A)=eig(\sigma_A)$ for any subsystem $A$. However, the states are not even LOCC comparable. In order to prove
that we
first note that the eigenvalues of $\rho_A$ and $\rho^\ast_A$
coincide for any subsystem $A$. Furthermore, since the state is symmetric
with respect to particle exchange, all single qubit reduced states coincide. They are of the form $\rho=\proj{+}+
2^{-n}(\sqrt{2}((e^{i\alpha}-1)\ket{1}\bra{+}+(e^{-i\alpha}-1)\ket{+}\bra{1})+|e^{i\alpha}-1|^2\ket{1}\bra{1})$,
with eigenvalues $1/2(1\pm 2^{-n}\sqrt{8-2^{2+n}+2^{2n}})$. Thus,
none of the reduced states is proportional to the identity and
therefore the states are LU--equivalent iff their standard forms
coincide. Since those states are LMESs, we know that their
trace decompositions are of the form $[H Z(\alpha)]^{\otimes
n}\ket{\Psi}$, where $\alpha$ is determined via the equation
$\cot(\alpha)=\frac{\langle X_1\rangle}{\langle
Y_1\rangle}=1-2^{n-1}$. It is straightforward to show that
$\ket{\Psi_{tr}}=e^{i\alpha
n/2}[\cos(\alpha/2)\ket{0}-i\sin(\alpha/2)\ket{1}]^{\otimes
n}+(i-1)e^{i\alpha n/2}[\ket{0}-\ket{1}]^{\otimes n}$ and therefore
none of the coefficients in the computational basis vanishes. There exist now several ways to prove that $\ket{\Psi}$ is not LU--equivalent to its complex conjugate. We could either compute the standard form, $\ket{\Psi_{s}}$, and show that it does not coincide with the one of $\ket{\Psi^\ast}$, i.e. show that $\ket{\Psi_{s}}$ is not real, or we could employ Lemma \ref{LemmaPhase} or Lemma \ref{LemmaPhase1}. We will use here Lemma \ref{LemmaPhase1} to show that the sorted trace decompositions are not related to each other by local phase gates, which proves that the states are not LU--equivalent. Since the first condition (condition (i)), $\vert \bra{\i}\Psi_{tr}\rangle\vert=\vert \bra{\i}\Psi^\ast_{tr}\rangle\vert$ is obviously fulfilled, we have that $\ket{\Psi}\simeq_{LU}\ket{\Psi^\ast}$ iff
$\ket{\Psi}/.\ket{\Psi^\ast}$ is a product state. Since
$\ket{\Psi}$ is symmetric, this last condition is fulfilled iff
there exists a single qubit state $\ket{\phi}$ such that
$\ket{\Psi}/.\ket{\Psi^\ast}=\ket{\phi}^{\otimes n}$. In other words,
the states are LU-equivalent iff there exists two phases $\alpha_0$
and $\alpha_1$ such that
$\ket{\Psi_{tr}}=e^{i\alpha_0}Z(\alpha_1)^{\otimes n}
\ket{\Psi^\ast_{tr}}$. This last equation is fulfilled iff there exists a phase $\alpha_1$ such that $
U_{1,\ldots,n}(\pi/2)\ket{+}^{\otimes n}= e^{i\alpha_0}V_{\alpha_1}^{\otimes
n} U_{1,\ldots,n}(-\pi/2)\ket{+}^{\otimes n}$, where
$V_{\alpha_1}=Z(-\alpha)HZ(\alpha_1)HZ(-\alpha)=\cos(\alpha_1)Z(-2\alpha)+i\sin(\alpha_1)e^{-i\alpha}X$.
Rewriting this condition we have that the states are LU--equivalent iff there exists a phase $\alpha_1$ such that $\ket{+}^{\otimes n}+(i-1)\ket{1}^{\otimes n}=
e^{i\alpha_0}((V_{\alpha_1}\ket{+})^{\otimes n}+ (-i-1)(V_{\alpha_1}\ket{1})^{\otimes
n})$. It can be easily shown that this condition can only be fulfilled if $V_{\alpha_1}\ket{+}\propto \ket{e_1}$ and
$V_{\alpha_1}\ket{+}\propto \ket{e_2}$, where $\ket{e_i}\in
\{\ket{1},\ket{+}\}$, for $i=1,2$ and $e_1\not \propto e_2$. We consider the two possible cases, (a) $V_{\alpha_1}\ket{+}=a\ket{1}$ for some $a\in \C$ and (b) $V_{\alpha_1}\ket{+}=a\ket{+}$ for some $a\in \C$. Case (a) is possible iff $\cos(\alpha_1)+i\sin(\alpha_1)e^{-i\alpha}=0$ and
 $\cos(\alpha_1)e^{-2i\alpha}+i\sin(\alpha_1)e^{-i\alpha}=a$. It is easy to see that the first condition cannot be fulfilled since $\alpha$ is determined as mentioned above. Since $V_{\alpha_1}\ket{+}=a\ket{+}$ (case (b)) implies that $\cos(\alpha_1)=0$ it is also easy to see that in this case $V\ket{1}=\pm ie^{-i\alpha}\ket{0}$ and therefore $V\ket{1}\neq a\ket{e_i}$ with $\ket{e_i}\in \{\ket{1},\ket{+}\}$, for $i=1,2$. This proves that for any $n>2$, the two $n$--qubit states, $\ket{\Psi}$ and $\ket{\Psi^\ast}$ are not LU--equivalent.

\section{LU--equivalence of mixed states and d--level systems}
\label{SecMixed}

We will show here that the criterion of LU--equivalence presented in \cite{Kr09} serves also as
a criterion of LU--equivalence for certain mixed and also for certain multipartite states which describe a system composed out of $d$--level systems.

For instance, if we want to find out whether or not two mixed states are
related to each other by local unitaries and if there exists at
least one eigenvalue of the mixed state which is not degenerate
then the same method can be used. This is due to the fact that the
unitaries cannot change the eigenvalues, thus, if
$\rho\simeq_{LU}\sigma$ then it must hold that
$\ket{\Psi}\simeq_{LU} \ket{\Phi}$, where $\ket{\Psi}$
($\ket{\Phi}$) denote the eigenstates to the non-degenerate eigenvalue of $\rho$, ($\sigma$) respectively. In order to check then if the two mixed states are LU--equivalent, one first uses the algorithm to determine the local unitaries which transform $\ket{\Phi}$ into $\ket{\Psi}$. Those unitaries must also transform $\sigma$ into $\rho$, which can then be easily checked.

The criterion for LU--equivalence for pure states can also be employed for mixed states if there does not exist a non--degenerate eigenvalue, but one which is two--fold degenerate. Let us denote by $\ket{\Psi_0}, \ket{\Psi_1}$ and $\ket{\Phi_0},\ket{\Phi_1}$, the eigenvectors corresponding to the two--fold degenerate eigenvalue  of $\rho$, $\sigma$ respectively. As before we have that if $\rho\simeq_{LU}\sigma$, then there exist local unitaries, $U_i$ such that $\proj{\Psi_0}+\proj{\Psi_1}=\bigotimes_i U_i (\proj{\Phi_0}+\proj{\Phi_1}) \bigotimes_i U^\dagger_i $.
This equation is fulfilled iff there exists a $2\times 2 $ unitary, $V$, with $\ket{\Psi_k}=\bigotimes_i U_i \sum_l V_{kl} \ket{\Phi_l}$, for $k=1,2$.
Note that this is equivalent to finding the local unitaries which map the state $\ket{\Phi}=\ket{0}\ket{\Phi_0}+\ket{1}\ket{\Phi_1}$ into
$\ket{\Psi}=\ket{0}\ket{\Psi_0}+\ket{1}\ket{\Psi_1}$. Thus, solving the LU--equivalence problem of mixed $n$--qubit states, where one eigenvalue is two--fold degenerate is equivalent to solving the LU--equivalence problem of $n+1$ qubit states, where $\rho_1=\one$.

Suppose now that $\rho$ is a $n$--qubit mixed state and its eigenvalue with the smallest degeneracy is $l$--fold degenerate. We denote by $\ket{\Psi_k}$ ($\ket{\Phi_k}$) the eigenstates of $\rho$, ($\sigma$) correspond to this eigenvalue. Then $\rho\simeq_{LU}\sigma$ implies that $\sum_k \ket{k}\ket{\Psi_k}=V \bigotimes_i U_i \sum_k \ket{k}\ket{\Phi_k}$, where $V$ is a $k\times k$ unitary and all the other $U_i$ are single qubit unitary operations. The idea is then to first fix the unitaries $U_i$ using the algorithm presented in \cite{Kr09} and at the end try to fix the unitary $V$.

The LU--equivalence problem for $d$--level systems can be investigated in a similar way. However, due to the additional degeneracy which can occur in this case, the situation gets more complicated. For instance, if a state $\ket{\Psi}$ describes a system composed out of $d$--level systems, then $\rho_i$, which describes a single $d$-level system can be $l$--fold degenerate, where $l\leq d$. Thus, in this case the unitaries occurring in Eq. (\ref{LU}) can in general not be determined up to local phase gates. If there is no degeneracy similar methods can of course be applied to solve the problem of LU--equivalence.

\section{Multipartite entanglement}

\label{SecEntanglement}

The algorithm presented in \cite{Kr09} cannot only be used to solve the LU--equivalence problem, but allows us also to gain a new insight into the entanglement properties of multipartite states.
Within the algorithm the classes $\rho_{n_1,\ldots n_l,k}\neq
\rho_{n_1,\ldots n_l}\otimes \one $ and $\rho_{n_1,\ldots n_l,k}=
\rho_{n_1,\ldots n_l}\otimes \one $, for some subset of qubits, $n_1,\ldots n_l,k$, are distinguished. In the first case the unitary $W_k$ can be computed as a function of the unitaries $U_{n_1},\ldots U_{n_l}$, whereas it cannot (using the proposed algorithm) be computed in the second case. Therefore, in this case a new variable, $U_k$ is required. As explained in \cite{Kr09} those classes correspond also to different entanglement classes. For instance, applying any von Neumann measurement on the first subsystem described by the state $\ket{\Psi}$, with $\rho_{12}=\rho_1\otimes \one$, always results in a state where the second system is maximally entangled with the remaining systems, independent of the measurement outcome. Obviously this is not the case for a state with $\rho_{12}\neq\rho_1\otimes \one$. This suggests that, in order to understand how a many--body system can be entangled, one first identifies the entanglement class (as described above) to which the state belongs to. Note that this classification is based on multipartite, not bipartite entanglement properties. However, it is easy to perform this classification since one only needs to consider the reduced state of certain subsystems. Within the identified entanglement class it is then feasible to understand how multipartite entanglement can be qualified and even quantified. For instance, as we have seen in Sec. \ref{secExamples}, the LU--equivalence classes of four--qubit states with $\rho_{ij}=\one$, for some systems $i$ and $j$ are characterized by three parameters. This is due to the fact that any state in this class (choosing w. l. o. g. $i=1,j=2$) can be written as $\ket{\Psi}=\one_{12}\otimes U_d\ket{\Phi^+}_{13}\ket{\Phi^+}_{24}$, where $U_d=e^{i(\phi_1 X\otimes X+\phi_2 Y\otimes Y+\phi_3 Z\otimes Z)}$, for some phases $\phi_i$. Thus, also the entanglement contained in such a state is completely characterized by $E_{12}(\ket{\Psi})=2$ and the three phases, $\phi_i$. Recall that any two--qubit gate, $U$, can be decomposed as $U=U_1\otimes U_2 U_d V_1\otimes V_2$, with $U_d$ as above denotes the non--local content of the gate $U$. 
Using all that allows us to give the three parameters $\phi_i$ the following physical meaning. Recall that the state $\ket{\Psi}=\one_{12} \otimes U_{34} \ket{\Phi^+}\ket{\Phi^+}$ is the Choi--Jamio\l kowski state corresponding to the operation $U$ \cite{CiDuKrLe00}. That is, given the state $\ket{\Psi}$ the operation $U$ can be implemented using just local operations \footnote{The implementation works as follows. First a system is prepared in the state $\ket{\Psi}$. Then local Bell--measurements are preformed on $\ket{\Psi}$ and the input state, $\rho$. In case the measurement $\ket{\Phi^+}$ is performed, the output state is ${\cal E}(\rho)$. In case any other measurement result is obtained, the output will be ${\cal E}(\Sigma_i\otimes \Sigma_j \rho \Sigma^\dagger_i\otimes \Sigma^\dagger_j)$, for properly chosen local operation $\Sigma_i\otimes \Sigma_j$.}. This shows that the non--local properties of a four--qubit state for which there exists a maximally entangled bipartite splitting between two versus two qubits is completely characterized by the amount of entanglement which can be generated using this state as the only non--local resource.

Of course, this new insight in characterizing multipartite entanglement by the amount of entanglement which can be generated using this state as the only non--local resource can be generalized to an arbitrary state, independent of the dimension and even for mixed states. Note that the quantum operation corresponding to a state describing $n$ subsystems is acting only on $\lceil n/2 \rceil$ systems. Like in the example of four--qubit states, the corresponding operation is acting on two qubits. This fact simplifies the characterization of multipartite entanglement, since e.g. the non--local properties of two--qubit operations are very well understood. It should be further noted here that the operation corresponding to the state, $\ket{\Psi}$ via the Choi--Jamio\l kowski isomorphism is unitary iff the state has the property that it is maximally entangled in the considered bipartite splitting.

Let us point out here that the algorithm gets more and more complicated the larger the number of systems, $l$ is for which $\rho_{n_1,\ldots n_l}=\one$ for any choice of $n_1,\ldots n_l$, since then $l$ unitaries have to be considered as variables. In the worst case, where any bipartition of
$\lceil n/2 \rceil$ qubits is maximally entangled with the rest, $\lceil n/2 \rceil$ unitaries have to be considered as variables.     It is known however, that only for very few values of $n$ such states
exist \cite{Bra03}. On the other hand, the more systems are maximally entangled with the rest, the less parameters remain to characterize the LU--equivalence class. Like, for instance in the example of four qubits, the class with $\rho_{ij}=\one$, for some systems $i,j$ can be characterized with only three parameters. For the other extreme case of generic states all the parameters occurring in the standard form
determine, like in the bipartite case, the entanglement contained
in the state.

Another important insight into multipartite entanglement which we gained here is the fact that for any $n$ there exists a $n$--qubit state, $\ket{\Psi}$, which is not LOCC comparable to its complex conjugate (Sec \ref{SecLOCC}). Thus, the non--local properties of $\ket{\Psi}$ and $\ket{\Psi^\ast}$ seem to be really different. Since the mapping $\ket{\Psi}\rightarrow \ket{\Psi^\ast}$
corresponds to the redefinition of the complex unit $i$ by $-i$, one might expect that this change does not lead to any new physics. In fact, for any observable $O$, we have $\bra{\Psi}O\ket{\Psi}= \bra{\Psi^\ast}O^\ast\ket{\Psi^\ast}$. Thus, whatever measurement outcome we can get by measuring a system described by the state $\ket{\Psi}$, the same outcome can be obtained by measuring $O^\ast$ on $\ket{\Psi^\ast}$. Due to that, there will not exist a physical measure which is capable of distinguishing those two states. This shows that it will not be possible to characterize all LU--equivalence classes by entanglement measures.

This suggest the introduction of a function, $I_1$, with $I_1(\ket{\Psi})=0$ if $\ket{\Psi}\simeq_{LU} \ket{\Psi^\ast}$ and $I_1(\ket{\Psi})=1$ otherwise. If $I_1(\ket{\Psi})=1$ the Hilbert space should be divided into two subsets, one containing $\ket{\Psi}$ and the other containing $\ket{\Psi^\ast}$. After making this distinction one proceeds investigating the non--local properties of the state $\ket{\Psi}$ within the subset associated to it.

In \cite{StCaKr10} we will follow the approach to investigate the multipartite entanglement properties in the way outlined here. In particular, we will consider multipartite state, describing several qubits, and will introduce the function which determines if a state is LU--equivalent to its complex conjugate or not. Moreover, we will analyze the entanglement contained in the state by investigating the amount of entanglement which can be generated using this state as the only non--local resource to implement quantum operations on a smaller system.

\section{conclusion}

We used the criterion of LU--equivalence of multipartite pure states to derive the different LU--equivalence classes of up to four qubits. For five--qubit states, which can be treated analogously, it is shown that the most complicated class of states, where all two--qubit reduced states are completely mixed can be easily considered using the algorithm developed in \cite{Kr09}. Even though it is in principle necessary to determine some of the local unitary operations as a function of some others, it is shown that for those cases this is, in fact, not required. That is, the unitaries can always be directly computed. The algorithm suggests to distinguish  different classes of entangled states, like the one where $\rho_{12}\neq\rho_1\otimes \one$ and where $\rho_{12}=\rho_1\otimes \one$. We considered here all the possible classes and showed that within certain classes new, operational entanglement parameters can be identified which completely characterize the non--local properties of the states. For instance, it has been shown that any four--qubit state for which one two--qubit
reduced state is completely mixed is LU--equivalent to a state
$\ket{\Psi}=\one \otimes U_d \ket{\Phi^+}\ket{\Phi^+}$, where $U_d=e^{i (\alpha_1 \x\otimes \x+ \alpha_2 \y\otimes \y +\alpha_3 \z\otimes \z)}$ with $\alpha_i\in \R$, is the non--local content of a two--qubit gate \cite{KrCi01}. The state $\ket{\Psi}$ is the Choi--Jamio\l kowski state corresponding to the operation $U_d$ \cite{CiDuKrLe00}. Thus, $U_d$ can be implemented by local operations if the state $\ket{\Psi}$ is used as a resource. This new approach of characterizing the entanglement of a multipartite state by the entangling capability of the operation which can be implemented using the state as the only non--local resource can be generalized to arbitrary states. Moreover, we derived examples of $n$--qubit states (for $n>2$) which are not LOCC comparable to their complex conjugate. This observation suggests the introduction of a new measure, which distinguishes the cases $\ket{\Psi}\simeq_{LU} \ket{\Psi^\ast}$ and $\ket{\Psi}\not\simeq_{LU} \ket{\Psi^\ast}$. If the states are not LU--equivalent, two different subsets of the Hilbert space should be considered, one for $\ket{\Psi}$ and one for $\ket{\Psi^\ast}$, in order to further investigate the properties of multipartite entangled states. In \cite{StCaKr10} we will prove that the examples of $n$--qubit states, $\ket{\Psi}$, which are presented here, cannot even be mapped into their complex conjugate by allowing stochastic LOCC (SLOCC). That is, it is not possible to transform $\ket{\Psi}$ into $\ket{\Psi^\ast}$ by local operations even with an arbitrary small probability of success.

\section{Acknowledgement}

The author would like to thank Geza Giedke for careful reading of the manuscript, and  Hans Briegel for continuous support
and interest in this work, and acknowledges support of the FWF
(Elise Richter Program).


\end{document}